%% file: rss-mfcs14.tex
\documentclass[usletter]{article}
\usepackage[totalwidth=450pt,totalheight=640pt]{geometry}
\usepackage{scrextend}
\usepackage[cmex10]{amsmath}
\usepackage{amssymb}
\usepackage{amsthm,todonotes}
\usepackage{graphicx}
\usepackage{microtype}
\usepackage{color}
\usepackage{hyperref}
\usepackage{enumitem}
\newtheorem{proposition}{Proposition}
\newtheorem{theorem}{Theorem}
\newtheorem{lemma}{Lemma}

\newtheorem{corollary}{Corollary}

\newcommand{\shortversion}[1]{}
\newcommand{\longversion}[1]{#1}
\newcommand{\longshort}[2]{\longversion{#1}\shortversion{#2}}

\usepackage{amsmath}
\usepackage{amssymb}
\usepackage{amsfonts,todonotes}
\usepackage{graphicx}
\usepackage[font={small}]{caption}
\usepackage{microtype}
\usepackage[rightcaption]{sidecap}    
\usepackage{enumitem}
\usepackage{hyperref}

\usepackage{floatrow}

\newcommand{\fo}{\mathcal{FO}}
\newcommand{\rela}{\mathbf{A}}

\newcommand{\pp}{\mathbf{P}}

\newcommand{\ldpt}{\mathrm{lower\textup{-}depth}}
\newcommand{\udpt}{\mathrm{upper\textup{-}depth}}
\newcommand{\dpt}{\mathrm{depth}}

\title{Quantified Conjunctive Queries\\on Partially Ordered Sets}

\begin{document}

\newtheorem{fact}{Fact}
\newtheorem{notation}{Notation}
\newtheorem{claimm}{Claim}

%

\longshort{
\newcommand*\samethanks[1][\value{footnote}]{\footnotemark[#1]}
\author{%
Simone Bova, Robert Ganian, and Stefan Szeider\\
\small Vienna University of Technology\\[-3pt]
\small  Vienna, Austria}

\date{}

}
{
\author{Simone Bova \and Robert Ganian \and Stefan Szeider}

\institute{Vienna University of Technology, Vienna, Austria}}
%

\maketitle
\begin{abstract}
We study the computational problem of checking whether a quantified conjunctive query 
(a first-order sentence built using only conjunction as Boolean connective) 
is true in a finite poset (a reflexive, antisymmetric, and transitive directed graph).  
We prove that the problem is already $\mathrm{NP}$-hard on a certain fixed poset, 
and investigate structural properties of posets yielding fixed-parameter tractability when 
the problem is parameterized by the query.  Our main algorithmic result is that model checking quantified conjunctive queries on posets 
of bounded width is fixed-parameter tractable (the width of a poset is the maximum size of a subset of pairwise incomparable elements).  
We complement our algorithmic result by complexity results 
with respect to classes of finite posets in a hierarchy of natural poset invariants, 
establishing its tightness in this sense. 
\end{abstract}


\section{Introduction}

%

\noindent \textit{Motivation.}  The \emph{model checking} problem for first-order logic 
is the problem of deciding whether a given first-order sentence is true in a given finite structure; 
it encompasses a wide range of fundamental combinatorial problems.  The problem is 
trivially decidable in $O(n^k)$ time, 
where $n$ is the size of the structure and $k$ is the size of the sentence, 
but it is not polynomial-time decidable 
or even \emph{fixed-parameter tractable} when parameterized by $k$ 
(under complexity assumptions in classical and parameterized complexity, respectively).  

Restrictions of the model checking problem to fixed classes of structures or sentences 
have been intensively investigated from the perspective of parameterized algorithms and complexity \cite{ChenDalmau12,Grohe07,GroheKreutzer11}. 
%
%
In particular, starting from seminal work by Courcelle \cite{Courcelle90recognizable} and Seese~\cite{Seese96}, 
structural properties of \emph{graphs} sufficient for fixed-parameter tractability of model checking 
have been identified.  An important outcome of this research is the understanding of the interplay between 
structural properties of graphs and the expressive power of first-order logic, most notably 
the interplay between sparsity and locality, 
culminating in the recent result by Grohe, Kreutzer, and Siebertz 
that model checking first-order logic on classes of 
\emph{nowhere dense} graphs is fixed-parameter tractable \cite{NesetrilOssonadeMendez12,GroheKreutzerSiebertz14}.  
On graph classes closed under subgraphs the result is known to be tight; 
at the same time, there are classes of \emph{somewhere dense} graphs (not closed under subgraphs) 
with fixed parameter tractable first-order (and even monadic second-order) logic model checking; 
the prominent examples are graph classes of bounded clique-width solved by Courcelle, Makowsky, and Rotics \cite{CourcelleMakowskyRotics00}.  

In this paper, we investigate \emph{posets} (short for \emph{partially ordered sets}).  
Posets form a fundamental class of combinatorial objects \cite{GrahamGrotschelLovasz95} and may be viewed as reflexive, 
antisymmetric, and transitive directed graphs.  Besides their naturality, 
our motivation towards posets is that 
they challenge our current model checking knowledge; 
indeed, posets are somewhere dense (but not closed under substructures) 
and have unbounded clique-width \cite[Proposition~5]{BovaGanianSzeider14}.  Therefore, 
not only are they not covered by the aforementioned results \cite{GroheKreutzerSiebertz14,CourcelleMakowskyRotics00}, 
but most importantly, it seems likely that new structural ideas and algorithmic techniques are needed 
to understand and conquer first-order logic on posets.  

In recent work, we started the investigation of first-order logic model checking on finite posets, 
and obtained a parameterized complexity classification of \emph{existential} and \emph{universal} logic 
(first-order sentences in prefix form built using only existential 
or only universal quantifiers) with respect to classes of posets 
in a hierarchy generated by basic poset invariants, including for instance width and depth \cite{BovaGanianSzeider14}.\footnote{Existential and universal logic are maximal syntactic fragments properly contained in first-order logic.}
In particular, as articulated more precisely in \cite{BovaGanianSzeider14}, 
a complete understanding of the first-order case 
reduces to understanding the parameterized complexity of model checking 
first-order logic on bounded width posets (the \emph{width} of a poset is the maximum size of a subset 
of pairwise incomparable elements); these classes are hindered by the same obstructions as general posets, 
since already posets of width $2$ have unbounded clique-width \cite[Proposition~5]{BovaGanianSzeider14}.

\medskip

\noindent \textit{Contribution.}  In this paper we push the tractability frontier 
traced in \cite{BovaGanianSzeider14} closer towards full first-order logic, 
by proving that model checking \emph{(quantified) conjunctive positive} logic 
(first-order sentences built using only conjunction as Boolean connective) is 
tractable on bounded width posets.\footnote{Conjunctive positive logic and existential (respectively, universal) logic are incomparable 
syntactic fragments of first-order logic.}  The problem of model checking conjunctive positive logic on finite structures, 
also known as the \emph{quantified constraint satisfaction} problem, 
has been previously studied with various motivations in various settings \cite{BornerBulatovChenJeavonsKrokhin09,ChenDalmau12}; 
somehow surprisingly, conjunctive logic is also capable of expressing rather interesting poset properties 
(as sampled in Proposition~\ref{prop:cpprop}).

More precisely, our contribution is twofold.  First, we identify conjunctive positive logic as a minimal syntactic fragment of first-order logic 
that allows for full quantification, and has computationally hard expression complexity on posets; 
namely, we prove that \emph{there exists a finite poset where 
model checking (quantified) conjunctive positive logic is $\mathrm{NP}$-hard} (Theorem~\ref{th:exprhard}).  
Next, as 
our main algorithmic result, 
we establish that \emph{model checking conjunctive positive logic on finite posets, 
parameterized by the width of the poset and the size of the sentence, is fixed-parameter tractable} 
with an elementary parameter dependence (Theorem~\ref{th:tractparamwidthsentence}).  
The aforementioned fact that model checking conjunctive positive logic is already $\mathrm{NP}$-hard on a fixed poset 
justifies the relaxation to fixed-parameter tractability by showing that, 
if we insist on polynomial-time algorithms, any structural property of posets 
(captured by the boundedness of a numeric invariant) is negligible.

Informally, the idea of our algorithm is the following.  
First, given a poset $\pp$ and a sentence $\phi$, 
we rewrite the sentence in a simplified form (which we call a \emph{reduced} form), 
equisatisfiable on $\pp$ (Proposition~\ref{prop:nf}).  
Next, using the properties of reduced forms, 
we define a syntactic notion of \lq\lq depth\rq\rq\ of a variable in $\phi$ 
and a semantic notion of \lq\lq depth\rq\rq\ of a subset of $\pp$, 
and we prove that $\pp \models \phi$ if and only if 
$\pp$ verifies $\phi$ upon \lq\lq relativizing\rq\rq\ variables to subsets of matching depth 
(Lemma~\ref{lemma:eloiserestr} and 
and Lemma~\ref{lemma:relativization}).  
The key fact is that 
the size of the subsets of $\pp$ used to relativize the variables of~$\phi$ 
is bounded above by the width of $\pp$ and the size 
of $\phi$ (Lemma~\ref{lemma:boundedsearch}), 
from which the main result follows (Theorem~\ref{th:tractparamwidthsentence}). 
We remark that the approach outlined above differs significantly from the algebraic approach used in \cite{BovaGanianSzeider14}; 
moreover, both stages make essential use of the restriction that conjunction is the only Boolean connective allowed in the sentences.

It follows immediately that 
\emph{model checking conjunctive positive logic on classes of finite posets of bounded width, 
parameterized by the size of the sentence, is fixed-parameter tractable} (Corollary~\ref{cor:bwtract}).  
On the other hand, there exist classes of finite posets of bounded depth 
(the \emph{depth} of a poset is the maximum size of a subset of pairwise comparable elements) 
and classes of finite posets of bounded cover-degree 
(the \emph{cover-degree} of a poset is the degree of its cover relation) 
where model checking conjunctive positive logic is shown to be $\textup{coW}[2]$-hard and hence 
not fixed parameter tractable, unless the exponential time hypothesis~\cite{FlumGrohe06} fails, see Proposition~\ref{prop:allhardnessresults}.  
Combined with the algorithm by Seese~\cite{Seese96}, these facts complete the parameterized complexity classification of the investigated poset invariants, 
as depicted in Figure~\ref{fig:classification}.  




\shortversion{\setcounter{figure}{0}
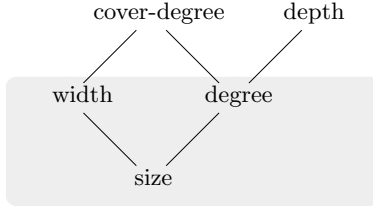
\begin{SCfigure}
\input{overwiewnotop.pspdftex}
\hfill 
\caption{On all classes of posets bounded under invariants 
in the gray region, model checking conjunctive positive logic 
is fixed-parameter tractable; on some classes of posets bounded 
under the remaining invariants, the problem 
is not fixed-parameter tractable unless $\textup{FPT}=\textup{coW}[2]$.}\label{fig:classification}
\end{SCfigure}}

\longversion{\setcounter{figure}{0}\begin{figure}
\input{overwiewnotop.pspdftex}
\caption{On all classes of posets bounded under invariants 
in the gray region, model checking conjunctive positive logic 
is fixed-parameter tractable; on some classes of posets bounded 
under the remaining invariants, the problem 
is not fixed-parameter tractable unless $\textup{FPT}=\textup{coW}[2]$.}\label{fig:classification}
\end{figure}}

The classification of conjunctive positive logic in this paper matches
the classification of existential logic in \cite{BovaGanianSzeider14}, 
and further emphasizes the quest for a classification of full first-order logic on bounded width posets.  
We believe that the work presented in this paper and \cite{BovaGanianSzeider14} 
enlightens the spectrum of phenomena that a fixed-parameter tractable algorithm for model checking 
the full first-order logic on bounded width posets, if it exists, has to capture.


\medskip

\noindent \shortversion{\emph{Throughout the paper, we mark with $\star$ all statements whose proofs 
are omitted; we refer to \cite{BovaGanianSzeiderIPEC14} for a full version.}}

\section{Preliminaries}\label{sect:prelim}

For all integers $k \geq 1$, we let $[k]$ denote the set $\{ 1, \ldots, k \}$.  
We focus on relational first-order logic.  
A \emph{vocabulary} $\sigma$ is a 
set of \emph{constant symbols} and \emph{relation symbols}; 
each relation symbol is associated to a natural number called its \emph{arity}; we let $\textup{ar}(R)$ denote the arity of $R \in \sigma$.  
\emph{All vocabularies considered in this paper are finite.} 

\longshort{An \emph{atom} $\alpha$ (over vocabulary $\sigma$) is an equality $t=t'$ 
or an application of a predicate $R t_1 \dots t_{\textup{ar}(R)}$, 
where $t,t',t_1,\dots,t_{\textup{ar}(R)}$ are variable symbols (in a fixed countable set) or constant symbols, and $R \in \sigma$.  
A \emph{formula} (over vocabulary $\sigma$) is built from atoms (over $\sigma$), 
conjunction ($\wedge$), disjunction ($\vee$), implication ($\to$), negation ($\neg$), 
universal quantification ($\forall$), and existential quantification ($\exists$).
A \emph{sentence} is a formula having no free variables. We let $\fo$ 
denote the class of first-order sentences.}{An \emph{atom} $\alpha$ (over vocabulary $\sigma$) is an equality $t=t'$ 
or an application of a predicate $R t_1 \dots t_{\textup{ar}(R)}$, 
where $t,t',t_1,\dots,t_{\textup{ar}(R)}$ are variable symbols (in a fixed countable set) or constant symbols, and $R \in \sigma$.  
We let $\fo$ denote the class of first-order sentences.}

A \emph{structure} $\rela$ (over $\sigma$) is specified by 
a nonempty set $A$, called the \emph{universe} of the structure, 
an element $c^{\rela} \in A$ for each constant symbol $c \in \sigma$, 
and a relation $R^{\rela} \subseteq A^{\textup{ar}(R)}$ for each relation symbol $R \in \sigma$.  
Given a structure $\mathbf{A}$ and $B \subseteq A$ such that 
$\{ c^\mathbf{A} \mid c \in \sigma \} \subseteq B$, 
we denote by $\mathbf{A}|_B$ the substructure of $\mathbf{A}$ induced by $B$, 
defined as follows: the universe of $\mathbf{A}|_B$ is $B$, 
$c^{\mathbf{A}|_B}=c^\mathbf{A}$ for each $c \in \sigma$, 
and $R^{\mathbf{A}|_B}=R^{\mathbf{A}} \cap B^{\textup{ar}(R)}$ 
for all $R \in \sigma$.  A structure is \emph{finite} if its universe is finite and \emph{trivial} if its universe is a singleton.  
\emph{All structures considered in this paper are finite and nontrivial.} 

For a structure $\rela$ and a sentence $\phi$ over the same vocabulary, 
we write $\rela \models \phi$ if the sentence $\phi$ is \emph{true} in the structure $\rela$.  
When $\rela$ is a structure, $f$ is a mapping from the variables to 
the universe of $\rela$, and $\psi(x_1,\ldots,x_n)$ is a formula over the vocabulary of $\rela$,
we write $\rela,f \models \psi$ or (liberally) $\rela \models \psi(f(x_1),\ldots,f(x_n))$ to indicate that $\psi$ is satisfied 
in $\rela$ under $f$.  \longversion{Let $\phi$ and $\psi$ be sentences over the same vocabulary $\sigma$.  
We say that $\phi$ \emph{entails} $\psi$ (denoted $\phi \models \psi$) if, for all structures $\rela$ over $\sigma$, it holds that $\rela \models \phi$ implies $\rela \models \psi$; 
we say that $\phi$ and $\psi$ are \emph{logically equivalent} 
(denoted $\phi \equiv \psi$) if $\phi \models \psi$ and $\psi \models \phi$.}  

%

\longversion{We refer the reader to \cite{FlumGrohe06} for 
the standard algorithmic setup of the model checking problem, 
including the underlying computational model, 
encoding conventions for input structures and sentences, 
and the notion of \emph{size} of the (encoding of an) input structure or sentence.  
We also refer the reader to \cite{FlumGrohe06} for further standard notions in parameterized complexity theory.

Here, we only recall that a \emph{parameterized problem} $(Q,\kappa)$ is a 
\emph{problem} $Q \subseteq \Sigma^*$ together 
with a \emph{parameterization} $\kappa \colon \Sigma^* \to \mathbb{N}$, where $\Sigma$ is a finite alphabet.  
A parameterized problem $(Q,\kappa)$ is \emph{fixed-parameter tractable (w.r.t.\   $\kappa$)}, 
in short \emph{fpt}, if there exists a decision algorithm for $Q$, 
a computable function $f \colon \mathbb{N} \to \mathbb{N}$, 
and a polynomial function $p \colon \mathbb{N} \to \mathbb{N}$, 
such that for all $x \in \Sigma^*$, the running time of the algorithm on $x$ 
is at most $f(\kappa(x)) \cdot p(|x|)$. 

The (parameterized) computational problem under consideration 
is the following.  Let $\sigma$ be a relational vocabulary, 
$\mathcal{C}$ be a class of $\sigma$-structures, 
and $\mathcal{L} \subseteq \fo$ be a class of $\sigma$-sentences.  
The \emph{model checking problem} for $\mathcal{C}$ and $\mathcal{L}$, 
in symbols $\textsc{MC}(\mathcal{C},\mathcal{L})$, is the problem of deciding, 
given $(\mathbf{A},\phi) \in \mathcal{C} \times \mathcal{L}$, 
whether $\mathbf{A} \models \phi$.  The parameterization, given an instance $(\mathbf{A},\phi)$, 
returns the size of the encoding of $\phi$.  \emph{In this paper, $\mathcal{C}$ is usually a class of partially ordered sets, 
and $\mathcal{L}$ is $\fo(\forall,\exists,\wedge)$.}  We let $\|(\mathbf{A},\phi)\|$, 
$\|\mathbf{A}\|$, and $\|\phi\|$ denote, respectively, the size of the instance $(\mathbf{A},\phi)$, the structure $\mathbf{A}$, 
and the sentence $\phi$.}
\shortversion{We refer the reader to \cite{FlumGrohe06} for 
the standard algorithmic setup of the model checking problem, 
and for standard notions in parameterized complexity theory.  
As for notation, the model checking problem for a class of $\sigma$-structures $\mathcal{C}$ 
and a class of $\sigma$-sentences $\mathcal{L} \subseteq \fo$ is denoted by $\textsc{MC}(\mathcal{C},\mathcal{L})$; 
it is the problem of deciding, given $(\mathbf{A},\phi) \in \mathcal{C} \times \mathcal{L}$, 
whether $\mathbf{A} \models \phi$.  We let $\|(\mathbf{A},\phi)\|$, 
$\|\mathbf{A}\|$, and~$\|\phi\|$ denote, respectively, the size of the (encoding of the) instance $(\mathbf{A},\phi)$, the structure $\mathbf{A}$, 
and the sentence $\phi$.  The parameterization of an instance $(\mathbf{A},\phi)$ 
returns $\|\phi\|$.}

\medskip

\noindent \textit{Conjunctive Positive Logic.} In this paper, we study the \emph{(quantified) conjunctive positive} fragment of first-order logic, 
in symbols $\fo(\forall,\exists,\wedge)$, 
containing first-order sentences built using only logical symbols in $\{\forall,\exists,\wedge\}$.

A conjunctive positive sentence is in \emph{alternating prefix form} if it has the form 
\begin{equation}\label{eq:strictform}
\phi=\forall x_1 \exists y_1 \ldots \forall x_l \exists y_l C(x_1,y_1,\ldots,x_l,y_l) \text{,} 
\end{equation}
where $l \geq 0$ and $C(x_1,y_1,\ldots,x_l,y_l)$ is a conjunction of atoms 
whose variables are contained in $\{x_1,y_1,\ldots,x_l,y_l\}$; it is possible to reduce 
any conjunctive positive sentence to a logically equivalent conjunctive positive 
sentence of form (\ref{eq:strictform}) in polynomial time.  For a simpler exposition, 
\emph{every conjunctive positive sentence considered in this paper 
is assumed to be given in alternating prefix form 
(or is implicitly reduced to that form if required by the context).}

Let $\sigma$ be a relational vocabulary.  Let $\rela$ be a $\sigma$-structure 
and let $\phi$ be a conjunctive positive $\sigma$-sentence as in (\ref{eq:strictform}).  
It is well known that the 
truth of $\phi$ in $\rela$ can be characterized in terms of the \emph{Hintikka (or model checking) game} 
on $\rela$ and $\phi$.  The game is played by two players, Abelard (male, 
the \emph{universal} player) and Eloise (female, the \emph{existential} player), 
as follows.  For increasing values of~$i$ from $1$ to $l$, 
Abelard assigns $x_i$ to an element $a_i \in A$, 
and Eloise assigns $y_i$ to an element $b_i \in A$; 
the sequence $(a_1,b_1,\ldots,a_l,b_l)$ is called a \emph{play} on $\rela$ and $\phi$, 
where $(a_1,\ldots,a_l)$ and $(b_1,\ldots,b_l)$ are the plays by Abelard and Eloise respectively;  
Eloise wins if and only if 
\shortversion{$\mathbf{A} \models C(a_1,b_1,\ldots,a_l,b_l)$.}
\longversion{$$\mathbf{A} \models C(a_1,b_1,\ldots,a_l,b_l)\text{.}$$}

A \emph{strategy for Eloise} (in the Hintikka game on $\rela$ and $\phi$) is a sequence $(g_1,\ldots,g_l)$ of functions 
of the form $g_i \colon A^i \to A$, for all $i\in [l]$; it \emph{beats} a play $f \colon \{x_1,\ldots,x_l\} \to A$ 
by Abelard 
\shortversion{if $\mathbf{A} \models C(f(x_1),g_1(f(x_1)),\ldots,f(x_i),g_i(f(x_1),\ldots,f(x_i)),\ldots)$, where $i \in [l]$.}
\longversion{if $$\mathbf{A} \models C(f(x_1),g_1(f(x_1)),\ldots,f(x_i),g_i(f(x_1),\ldots,f(x_i)),\ldots)\text{,}$$ where $i \in [l]$.}
A strategy for Eloise is \emph{winning} (in the Hintikka game on $\rela$ and $\phi$) if it beats all Abelard plays.  
It is well known (and easily verified) that $\rela \models \phi$ if and only if 
Eloise has a winning strategy (in the Hintikka game on $\rela$ and $\phi$).  

\longversion{For $X_1,Y_1,\ldots,X_l,Y_l \subseteq A$, we freely denote by 
\begin{equation}\label{eq:relatstrictform}
\phi'=(\forall x_1 \in X_1)(\exists y_1 \in Y_1)\ldots(\forall x_l \in X_l)(\exists y_l \in Y_l)C(x_1,y_1,\ldots,x_l,y_l) \text{,} 
\end{equation}
the relativization in $\phi$ of variable $x_i$ to $X_i$ and $y_i$ to $Y_i$ for all $i \in [l]$.  
We liberally write $\rela \models \phi'$ to mean that $\rela^* \models \phi^*$, 
where $\rela^*$ and $\phi^*$ have vocabulary $\sigma \cup \{X_1,Y_1,\ldots,X_l,Y_l\}$, 
the $\sigma$-reduct of $\rela^*$ is equal to $\rela$, 
$X_i^{\rela^*}=X_i$ and $Y_i^{\rela^*}=Y_i$ for all $i \in [l]$, 
and 
\begin{equation*}
\phi^*=\forall x_1 \exists y_1 \ldots \forall x_l \exists y_l( X_1x_1 \to( Y_1y_1 \wedge ( \cdots X_lx_l \to( Y_ly_l \wedge C ) \cdots ) ) ) \text{.} 
\end{equation*}}
\shortversion{For $X_1,Y_1,\ldots,X_l,Y_l \subseteq A$, we denote by 
\begin{equation}\label{eq:relatstrictform}
\phi'=(\forall x_1 \in X_1)(\exists y_1 \in Y_1)\ldots(\forall x_l \in X_l)(\exists y_l \in Y_l)C(x_1,y_1,\ldots,x_l,y_l) \text{,} 
\end{equation}
the relativization in $\phi$ of variable $x_i$ to $X_i$ and $y_i$ to $Y_i$ for all $i \in [l]$, 
and liberally write $\rela \models \phi'$ meaning that $\phi'$ is satisfied in the intended expansion of $\rela$.}  
It is readily verified that, if $\phi'$ is as in (\ref{eq:relatstrictform}), 
then $\rela \models \phi'$ if and only if, in the Hintikka game on $\rela$ and $\phi$,  
Eloise has a strategy of the form $g_i \colon X_1 \times \cdots \times X_i \to Y_i$ for all $i \in [l]$, 
beating all plays $f$ by Abelard such that $f(x_i) \in X_i$ for all $i \in [l]$. 

\longversion{\medskip

\noindent \textit{Partially Ordered Sets.} A structure $\mathbf{G}=(V,E^\mathbf{G})$ 
with $\textup{ar}(E)=2$ is called a \emph{digraph}.  Two digraphs $\mathbf{G}$ 
and $\mathbf{H}$ are \emph{isomorphic} if there exists a bijection $f \colon G \to H$ 
such that for all $g,g' \in G$ it holds that $(g,g') \in E^\mathbf{G}$ 
if and only if $(f(g),f(g')) \in E^\mathbf{H}$.  

Let $\mathbf{G}$ be a digraph.  The \emph{degree} of $g \in G$, in symbols $\textup{degree}(g)$,  
is equal to $|\{ (g',g) \in E^\mathbf{G} \mid g' \in G \} 
\cup \{ (g,g') \in E^\mathbf{G} \mid g' \in G \}|$, 
and the \emph{degree} of $\mathbf{G}$, in symbols $\textup{degree}(\mathbf{G})$, 
is the maximum degree attained by the elements of $\mathbf{G}$.

A digraph $\pp=(P,\leq^\pp)$ is a \emph{partially ordered set} (in short, a \emph{poset}) if $\leq^\pp$ is a \emph{reflexive}, 
\emph{antisymmetric}, and \emph{transitive} relation over $P$, 
that is, respectively, $\pp \models \forall x(x \leq x)$, 
$\pp \models \forall x \forall y((x \leq y \wedge y \leq x) \to x=y)$, 
and $\pp \models \forall x \forall y \forall z((x \leq y \wedge y \leq z) \to x \leq z)$.  

An element $p \in P$ is \emph{minimal} if $\pp \models \forall x(x \leq p \to x=p)$, 
and \emph{maximal} if $\pp \models \forall x(x \geq p \to x=p)$.  For all $Q \subseteq P$, 
we let $\mathrm{min}^\pp(Q)$ and $\mathrm{max}^\pp(Q)$ denote, respectively, 
the set of minimal and maximal elements in the substructure of $\pp$ 
induced by $Q$; we also write $\mathrm{min}(\pp)$ instead of $\mathrm{min}^\pp(P)$, 
and $\mathrm{max}(\pp)$ instead of $\mathrm{max}^\pp(P)$.  An element $b \in P$ 
such that $\pp \models \forall x(b \leq x)$ is called the \emph{bottom} of $\pp$, 
and similarly an element $t \in P$ such that $\pp \models \forall x(x \leq t)$ is called the \emph{top} of $\pp$ 
(uniqueness of top and bottom, if they exists, is clear).

For all $Q \subseteq P$, 
we let $(Q]^{\mathbf{P}}$ denote the \emph{downset} of $\mathbf{P}$ induced by $Q$, that is, 
$(Q]^{\mathbf{P}}=\{ p \in P \mid \text{there exists $q \in Q$ such that $p \leq^{\mathbf{P}} q$}\}$; if $Q=\{q\}$, we also write $(q]^\mathbf{P}$ 
instead of $(\{q\}]^\mathbf{P}$.  The \emph{upset} of $\mathbf{P}$ induced by $Q$ is defined dually, 
$[Q)^{\mathbf{P}}=\{ p \in P \mid \text{there exists $q \in Q$ such that $p \geq^{\mathbf{P}} q$}\}$.  

Let $\pp$ be a poset and let $p,q \in P$.  We say that $q$ \emph{covers} $p$ in $\pp$ (denoted $p \prec^{\pp} q$) 
if $p<^{\pp}q$ and, for all $r \in P$, $p \leq^{\pp} r <^{\pp}q$ implies $p=r$.  We say that $p$ and $q$ are \emph{incomparable} 
in $\pp$ (denoted $p \parallel^{\pp} q$) if 
$\pp \not\models p \leq q \vee q\leq p$.  
The \emph{cover graph} of $\pp$ is the digraph $\textup{cover}(\pp)$ with vertex set $P$ 
and edge set $\{ (p,q) \mid p \prec^{\pp} q \}$.  It is well known that 
computing the cover relation corresponding to a given order relation, 
and vice versa the order relation corresponding to a given cover relation, 
is feasible in polynomial time \cite{Schroder03}.  If $\mathcal{P}$ 
is a class of posets, we let $\textup{cover}(\mathcal{P})=\{ \textup{cover}(\pp) \mid \pp \in \mathcal{P} \}$. 

In this paper, a poset $\pp$ is pictorially 
represented by its \emph{Hasse diagram}, that is a planar drawing of $\textup{cover}(\pp)$ 
where all edges are oriented upwards (thus, in the actual drawing, orientations are neglected).  
 
A \emph{chain} in $\pp$ is a subset $C \subseteq P$ 
such that $p \leq^{\pp} q$ or $q \leq^{\pp} p$ for all $p,q \in C$; 
in particular, if $P$ is a chain in $\pp$, 
we call $\pp$ itself a chain.  
An \emph{antichain} in $\pp$ is a subset $A \subseteq P$ 
such that $p \parallel^{\pp} q$ for all $p,q \in A$; 
in particular, if $P$ is an antichain in $\pp$, 
we call $\pp$ itself an antichain.  

Let $\mathcal{P}$ be the class of all posets.  A \emph{poset invariant} 
is a mapping $\textup{inv} \colon \mathcal{P} \to \mathbb{N}$ such that $\textup{inv}(\mathbf{P})=\textup{inv}(\mathbf{Q})$ 
for all $\mathbf{P},\mathbf{Q} \in \mathcal{P}$ where $\mathbf{P}$ and $\mathbf{Q}$ are isomorphic.  
Let $\textup{inv}$ be any invariant over $\mathcal{P}$.  Let $\mathcal{P}$ be any class of posets.  
We say that $\mathcal{P}$ is \emph{bounded} w.r.t.\   $\textup{inv}$ if there exists $b\in \mathbb{N}$ such that 
$\textup{inv}(\mathcal{P})\leq b$ for all $\mathbf{P} \in \mathcal{P}$.  Two poset invariants 
$\textup{inv}$ and $\textup{inv}'$ are naturally ordered by stipulating 
that $\textup{inv} \leq \textup{inv}'$ if and only if 
for every class $\mathcal{P}$ of posets, if $\mathcal{P}$ is bounded w.r.t.\   $\textup{inv}$, 
then $\mathcal{P}$ is bounded w.r.t.\   $\textup{inv}'$.

We introduce a family of natural poset invariants. Let $\pp$ be a poset.  
The \emph{size} of $\pp$ is the cardinality of its universe, $|P|$.   
The \emph{depth} of $\pp$, in symbols $\textup{depth}(\pp)$, 
is the maximum size attained by a chain in $\pp$.  
The \emph{width} of $\pp$, in symbols $\textup{width}(\pp)$, 
is the maximum size attained by an antichain in $\pp$.  
The \emph{degree} of $\pp$, in symbols $\textup{degree}(\pp)$, 
is the degree of $\pp$ as a digraph.  
The \emph{cover-degree} of $\pp$, 
in symbols $\textup{cover\textup{-}degree}(\pp)$, is the degree of the cover relation of $\pp$, 
that is, $\textup{degree}(\textup{cover}(\pp))$.  In \cite[Proposition~3]{BovaGanianSzeider14}, 
we prove that such poset invariants are ordered as in Figure~\ref{fig:classification}.}

\section{Expression Hardness}\label{sect:m1}

In this section we prove that conjunctive positive logic on posets 
is $\mathrm{NP}$-hard in expression complexity.  Let $\mathbf{B}=(B,\leq^\mathbf{B})$ be the \emph{bowtie} poset 
defined by the universe $B=\{0,1,2,3\}$ and the covers $0,2 \prec^\mathbf{B} 1,3$; see 
Figure~\ref{fig:exprhard}.



\begin{theorem}\label{th:exprhard}
$\textsc{MC}(\{\mathbf{B}\},\mathcal{FO}(\forall,\exists,\wedge))$ is $\textup{NP}$-hard.
\end{theorem}


\newcommand{\pfclnpharda}[0]{
\begin{proof} 
Let $f$ be a trivial assignment.  
Then, $\{f(y_0),f(y_2)\} \neq \{0,2\}$ or $\{f(y_1),f(y_3)\} \neq \{1,3\}$.  
We prove that the statement holds if $\{f(y_0),f(y_2)\} \neq \{0,2\}$; 
the case $\{f(y_1),f(y_3)\} \neq \{1,3\}$ is symmetric.

We distinguish two subcases.  First, assume that  $f(y_0)=b \in \{0,2\}$.  Since $\{f(y_0),f(y_2)\} \neq \{0,2\}$, 
either $f(y_2)=b$ or $f(y_2) \in \{1,3\}$.  Hence $b \leq^{\mathbf{B}} f(y_2)$.  For $i \in \{1,3\}$, let $b_i$ be any element of $B$ 
such that $b, f(y_i) \leq^\mathbf{B} b_i$; note that such a $b_i$ 
exists by construction of $\mathbf{B}$.  We now define the required extension $f'$ of $f$ 
by putting $f'(x_0)=f'(x_2)=f'(w_0)=f'(w_2)=f'(w_1)=f'(w_3)=b$, 
$f'(x_1)=b_1$, and $f'(x_3)=b_3$.  By inspection of $\mathbf{B}$, 
it holds that $\mathbf{B},f' \models \alpha$, and the first subcase is settled.

Second, assume that $f(y_0)=b \in \{1,3\}$.  By construction of $\mathbf{B}$, 
there exists $b_2 \in \{0,2\}$ such that $b_2 \leq^\mathbf{B} f(y_2), b$; such a $b_2$ is unique 
if $f(y_2) \in \{0,2\}$.  Moreover, there exist 
$b_1,b_3 \in B$ such that $b_2,f(y_1) \leq^\mathbf{B} b_1$ 
and $b_2,f(y_3) \leq^\mathbf{B} b_3$.  The required extension $f'$ of $f$ is defined 
by letting $f'(x_0)=f'(x_2)=f'(w_0)=f'(w_2)=f'(w_1)=f'(w_3)=b_2$, 
$f'(x_1)=b_1$, and $f'(x_3)=b_3$.  The second subcase is settled, 
and the claim is proved.
\shortversion{\qed}\end{proof}}

\newcommand{\pfclnphardb}[0]{
\begin{proof} 
Let $f$ be a trivial assignment.  By Claim~\ref{cl:nphard1}, 
let $f'$ be an extension of $f$ to the variables of $\alpha$ such that $\mathbf{B},f' \models \alpha$ and $f'(w_0)=f'(w_2)=f'(w_1)=f'(w_3)=b\in B$.  
Since $\psi'$ is a conjunction of atoms of the form $u \leq u'$, or $w_i \leq u$, or $u \leq w_i$ 
(where $u$ and $u'$ are variables not occurring in $\alpha$, and $i \in B$), 
and since $\leq^\mathbf{B}$ is reflexive, any assignment $f''$ that extends $f'$ by assigning 
all variables in $\psi'$ not occurring in $\alpha$ to $b$ 
is such that $\mathbf{B},f'' \models \alpha \wedge \psi'$, which settles the claim.
\shortversion{\qed}\end{proof}}

\newcommand{\pfclnphardc}[0]{
\begin{proof} 
Let $f$ be a nontrivial assignment, 
say $f(y_i)=b_i$ for all $i \in B$, $\{b_0,b_2\}=\{0,2\}$ and $\{b_1,b_3\}=\{1,3\}$.  We prove the two statements.

$(i)$  Clearly the extension $f'$ of $f$ defined by $f'(x_i)=f'(w_i)=b_i$ for all $i \in B$ 
verifies $\mathbf{B},f' \models \alpha$.

$(ii)$ Let $f'$ be any extension of $f$ to the variables of $\alpha$ such that $\mathbf{B},f' \models \alpha$.  We prove 
that $f'(x_i)=f'(w_i)=b_i$ for all $i \in B$, which suffices.  For $i \in \{0,2\}$, 
the atom $x_i \leq y_i$ in $\alpha$ forces $f'(x_i)=f'(y_i)=b_i$ because $b_i$ is minimal in $\mathbf{B}$.  
So $\{f'(x_0),f'(x_2)\}=\{0,2\}$.  
Similarly, for $i \in \{1,3\}$, the atom $y_i \leq x_i$ forces $f'(x_i)=f'(y_i)=b_i$ because $b_i$ is maximal in $\mathbf{B}$.  
So $\{f'(x_1),f'(x_3)\}=\{1,3\}$.  
For $i \in \{0,2\}$, the atoms $x_i \leq w_i \leq w_1 \leq x_1$ and $x_i \leq w_i \leq w_1 \leq x_3$ 
force $f'(w_i)=f'(x_i)=b_i$ because $b_i$ is the unique element in $B$ below both $1$ and $3$.  Similarly, for $i \in \{1,3\}$, the atoms $x_0 \leq w_0 \leq w_i \leq x_i$ and $x_2 \leq w_2 \leq w_i \leq x_i$ 
force $f'(w_i)=f'(x_i)=b_i$ because $b_i$ is the unique element in $B$ above $0$ and $2$ 
and below $b_i \in \{1,3\}$.
\shortversion{\qed}\end{proof}}

\newcommand{\pfclnphardd}[0]{
\begin{proof} 
Let $f$ be a nontrivial assignment.

$(i) \Rightarrow (ii)$ Assume $\mathbf{B},f \models \exists x_0 \ldots \exists x_3\exists w_0 \ldots \exists w_3( \alpha \wedge \psi' )$.  
Therefore, there exists an assignment $f'$ extending $f$ such that 
$\mathbf{B},f' \models \alpha \wedge \psi'$.  In particular, 
$\mathbf{B},f' \models \alpha$, hence by Claim~\ref{cl:nphard3}$(ii)$, 
it holds that $\{f'(w_0),f'(w_2)\}=\{0,2\}$ and $\{f'(w_1),f'(w_3)\}=\{1,3\}$; 
in particular, $f'$ restricted to $\{w_0,w_1,w_2,w_3\}$ is bijective into $B$.  
Let $u_1,\ldots,u_n$ be the variables of $\psi$.  We let the assignment 
$g \colon \{u_1,\ldots,u_n\} \to B^*$ be the unique mapping satisfying the following: 
for all $u \in \{u_1,\ldots,u_n\}$ and $i \in B^*$, 
$$g(u)=i \text{ if and only if } f'(u)=f'(w_i)\text{.}$$
note that such a unique $g$ exists by the properties of $f'$.

We check that $g$ witnesses $\mathbf{B}^* \models \psi$.  
For $i \in B^*$ and any variable $u$, let the atom $c_i \leq u$ be in $\psi$ 
(the argument is similar for an atom of the form $u \leq c_i$ in $\psi$).  
The atom $w_i \leq u$ is in $\psi'$ by construction, 
hence $f'(w_i) \leq^\mathbf{B} f'(u)$ by hypothesis.  
If $f'(w_i)=f'(u)$, then $g(u)=i$, and we are done since $c^{\mathbf{B}^*}_i=i \leq^{\mathbf{B}^*} g(u)$.  
If $f'(w_i) <^\mathbf{B} f'(u)$, then $i \in \{0,2\}$ 
and $f'(u)=f'(w_j)$ for some $j \in \{1,3\}$ by the properties of $f'$ 
and by inspection of $\mathbf{B}$.  It follows 
that $g(u) \in \{1,3\}$, and we are done since $c^{\mathbf{B}^*}_i \in \{0,2\}$ 
and $\{0,2\} \leq^{\mathbf{B}^*} \{1,3\}$.  

Consider variables $u$ and $u'$ such that the atom $u \leq u'$ is in $\psi$.  
By construction, the atom $u \leq u'$ is in $\psi'$, 
hence $f'(u) \leq^{\mathbf{B}} f'(u')$.  Let $i,j \in B$ 
such that $f'(w_i)=f'(u) \leq^{\mathbf{B}} f'(u')=f'(w_j)$; 
note that such $i$ and $j$ exist by the properties of $f'$.  
It follows that $g(u)=i$ and $g(u')=j$.  We now claim that 
$i \leq^{\mathbf{B}^*} j$.  Indeed, we have $f'(w_i) \leq^{\mathbf{B}} f'(w_j)$.  
If $f'(w_i)=f'(w_j)$, then $i=j$ by the properties of $f'$, and we are done since $g(u)=g(u')$.  
If $f'(w_i) <^{\mathbf{B}} f'(w_j)$, then $i \in \{0,2\}$ and $j \in \{1,3\}$, 
and as above, we are done since $g(u) \in \{0,2\}$, $g(u') \in \{1,3\}$, 
and $\{0,2\} \leq^{\mathbf{B}^*} \{1,3\}$.

$(ii) \Rightarrow (i)$ Let $g$ be any assignment witnessing $\mathbf{B}^* \models \psi$.  
Let $f$ be any assignment of $\{y_0,y_1,y_2,y_3\}$ in $B$.  If $f$ is trivial, 
then $\mathbf{B},f \models \exists x_0 \ldots \exists x_3\exists w_0 \ldots \exists w_3( \alpha \wedge \psi' )$ 
by Claim~\ref{cl:nphard2}.  Otherwise, assume that $f$ is nontrivial.  By Claim~\ref{cl:nphard3}$(i)$, 
let $f'$ be an extension of $f$ to the variables of $\alpha$ such that $\mathbf{B},f' \models \alpha$.  By Claim~\ref{cl:nphard3}$(ii)$, 
it holds that $\{f'(w_0),f'(w_2)\}=\{0,2\}$ and $\{f'(w_1),f'(w_3)\}=\{1,3\}$.  Let $f''$ 
extend $f'$ to the variables of $\psi'$ by putting, 
for all $u \in \{u_1,\ldots,u_n\}$ and $i \in B^*$:
$$f''(u)=f'(w_i) \text{ if and only if } g(u)=i \text{.}$$
It suffices to show that $\mathbf{B},f'' \models \psi'$.  

For $i \in B^*$ and $u$ a variable, $w_i \leq u$ be in $\psi'$ 
(atoms $c_i \leq u$ in $\psi'$ are similarly addressed).  
Then the atom $c_i \leq u$ is in $\psi$ by construction.  
Then $i \leq^{\mathbf{B}^*} g(u)$ by hypothesis.  
If $i=g(u)$, then $f''(u)=f'(w_i)=f''(w_i)$, and we are done.  
If $c^{\mathbf{B}^*}_i=i <^{\mathbf{B}^*} g(u)$, then $i \in \{0,2\}$ 
and $g(u)=j \in \{1,3\}$, that is, $f'(w_i) \in \{0,2\}$ and $f''(u)=f'(w_j) \in \{1,3\}$, 
from which $f''(w_i) \leq^{\mathbf{B}} f''(u)$ and we are done.

For $u$ and $u'$ variables, let the atom $u \leq u'$ be in $\psi'$.  
By construction, the atom $u \leq u'$ is in $\psi$, 
hence $i=g(u) \leq^{\mathbf{B}^*} g(u')=j$.  If $i=j$, 
then $f''(u)=f'(w_i)=f'(w_j)=f''(u')$, and we are done.  
If $i <^{\mathbf{B}^*} j$, then $i \in \{0,2\}$ and $j \in \{1,3\}$, 
then $f''(u)=f'(w_i) \in \{0,2\}$ and $f''(u')=f'(w_j) \in \{1,3\}$, 
from which $f''(u) \leq^{\mathbf{B}} f''(u')$ and we are done.
\shortversion{\qed}\end{proof}}

\begin{proof}
Let $\tau=\{\leq\}$ and $\sigma=\tau \cup \{c_0,c_1,c_2,c_3\}$ be vocabularies 
where $\leq$ is a binary relation symbol and $c_i$ is a constant symbol ($i \in B$).  
Let $\mathcal{FO}_{\sigma}(\exists,\wedge)$ contain first-order sentences built using only logical symbols in $\{\exists,\wedge\}$ 
and nonlogical symbols in $\sigma$; $\mathcal{FO}_{\tau}(\forall,\exists,\wedge)$ is described similarly.  
Let $\mathbf{B}^*$ be the $\sigma$-structure 
such that $B^*=B$, $(B^*,\leq^{\mathbf{B}^*})$ is isomorphic to $\mathbf{B}$ under the identity mapping, 
and $c_i^{\mathbf{B}^*}=i$ for all $i \in B$.  

By \cite[Theorem~2, Case $n=2$]{PrattTiuryn96}, 
the problem $\textsc{MC}(\{\mathbf{B}^*\},\mathcal{FO}_{\sigma}(\exists,\wedge))$ is $\textup{NP}$-hard.  
It is therefore sufficient to 
give a polynomial-time many-one reduction from $\textsc{MC}(\{\mathbf{B}^*\},\mathcal{FO}_{\sigma}(\exists,\wedge))$  
to $\textsc{MC}(\{\mathbf{B}\},\mathcal{FO}_{\tau}(\forall,\exists,\wedge))$.  
The idea of the reduction is to simulate the constants in $\sigma$ 
by universal quantification and additional variables; the details follow. 


Let $\psi$ be an instance of $\textsc{MC}(\{\mathbf{B}^*\},\mathcal{FO}_{\sigma}(\exists,\wedge))$, 
and let $\{x_i,y_i,w_i \mid i \in B\}$ be a set of $12$ fresh variables (not occurring in $\psi$).  
Let $\psi'$ be the $\mathcal{FO}_{\tau}(\exists,\wedge)$-sentence obtained from $\psi$ by replacing 
atoms of the form $c_i \leq u$ and $u \leq c_i$, respectively, 
by atoms of the form $w_i \leq u$ and $u \leq w_i$ (where $c_i$ is a constant in $\sigma$ and $u, w_i$ are variables).  
Let $\alpha$ be the conjunction of atoms defined by (see Figure~\ref{fig:exprhard})
$$\{w_0,w_2\} \leq \{w_1,w_3\} \wedge \bigwedge_{j \in \{0,2\}}\{x_j\} \leq \{y_j, w_j\} \wedge \bigwedge_{j \in \{1,3\}}\{y_j,w_j\} \leq \{x_j\}\text{,}$$
where, for sets of variables $S$ and $S'$, 
the notation $S \leq S'$ denotes the conjunction of atoms of the form $s \leq s'$ for all $(s,s') \in S \times S'$.  

\shortversion{\begin{SCfigure}[\sidecaptionrelwidth][t]
\input{exprhard.pspdftex}
\hfill 
\caption{The Hasse diagrams of the bowtie poset $\mathbf{B}$ (left) and of the representation $\mathbf{M}_\alpha$ of the formula $\alpha$ (right, 
see Section~\ref{sect:red} for the interpretation of $\mathbf{M}_\alpha$) used in Theorem~\ref{th:exprhard}. The idea of the reduction 
is to simulate the constant $c_i$ in $\psi \in \mathcal{FO}_{\sigma}(\exists,\wedge)$, interpreted on the element $i \in B$, 
by the variable $w_i$ in $\phi \in \mathcal{FO}_{\tau}(\forall,\exists,\wedge)$, where $i \in \{0,1,2,3\}$.}\label{fig:exprhard}
\end{SCfigure}}

\longversion{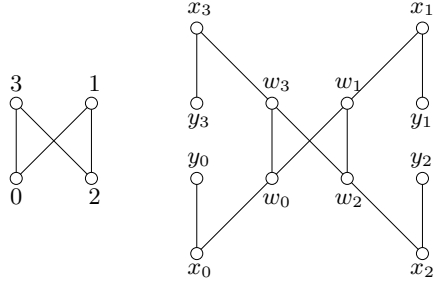
\begin{figure}
\input{exprhard.pspdftex}
\caption{The Hasse diagrams of the bowtie poset $\mathbf{B}$ (left) and of the representation $\mathbf{M}_\alpha$ of the formula $\alpha$ (right, 
see Section~\ref{sect:red} for the interpretation of $\mathbf{M}_\alpha$) used in Theorem~\ref{th:exprhard}. The idea of the reduction 
is to simulate the constant $c_i$ in $\psi \in \mathcal{FO}_{\sigma}(\exists,\wedge)$, interpreted on the element $i \in B$, 
by the variable $w_i$ in $\phi \in \mathcal{FO}_{\tau}(\forall,\exists,\wedge)$, where $i \in \{0,1,2,3\}$.}\label{fig:exprhard}
\end{figure}}

We finally define 
the $\mathcal{FO}_{\tau}(\forall,\exists,\wedge)$-sentence $\phi$ by putting 
\shortversion{$\phi=\forall y_0 \ldots \forall y_3\exists x_0 \ldots \exists x_3\exists w_0 \ldots \exists w_3( \alpha \wedge \psi' )$.}
\longversion{$$\phi=\forall y_0 \ldots \forall y_3\exists x_0 \ldots \exists x_3\exists w_0 \ldots \exists w_3( \alpha \wedge \psi' )\text{.}$$}
The reduction is clearly feasible in polynomial time; 
we now prove that the reduction is correct, that is, 
$\mathbf{B}^* \models \psi$ if and only if $\mathbf{B} \models \phi$.  

An assignment $f \colon \{y_0,y_1,y_2,y_3 \} \to B$ is said to be \emph{nontrivial} 
if $\{f(y_0),f(y_2)\}=\{0,2\}$ and $\{f(y_1),f(y_3)\}=\{1,3\}$, 
and \emph{trivial} otherwise; in particular, 
nontrivial assignments are bijective.


\longversion{
\begin{claimm}
\label{cl:nphard1}
Let $f$ be a trivial assignment.  There exists an assignment $f'$ extending $f$ 
to the variables of $\alpha$ such that $\mathbf{B},f' \models \alpha$ 
and $f'(w_0)=f'(w_2)=f'(w_1)=f'(w_3)$.
\end{claimm} 
\longversion{\pfclnpharda}}

\longshort{\begin{claimm}}{\begin{claimm}[$\star$]}
\label{cl:nphard2}
\longversion{$\mathbf{B},f \models \exists x_0 \ldots \exists x_3\exists w_0 \ldots \exists w_3( \alpha \wedge \psi' )$ for all trivial assignments $f$.}
\shortversion{$\mathbf{B},f \models \exists x_0 \ldots x_3 w_0 \ldots w_3( \alpha \wedge \psi' )$ for all trivial assignments $f$.}
\end{claimm}

\longversion{\pfclnphardb}


\longversion{\begin{claimm}
\label{cl:nphard3}
Let $f$ be a nontrivial assignment.  The following statements hold.
\begin{enumerate}[label=\textit{(\roman*)}]
\item There exists an assignment $f'$ extending $f$ to the variables of $\alpha$ such that $\mathbf{B},f' \models \alpha$.
\item For all assignments $f'$ extending $f$ to the variables of $\alpha$ such that $\mathbf{B},f' \models \alpha$, 
it holds that $\{f'(w_0),f'(w_2)\}=\{0,2\}$ and $\{f'(w_1),f'(w_3)\}=\{1,3\}$.
\end{enumerate}
\end{claimm}
\longversion{\pfclnphardc}}

\longshort{\begin{claimm}}{\begin{claimm}[$\star$]}
\label{cl:nphard4}
Let $f$ be a nontrivial assignment.  The following are equivalent.
\begin{enumerate}[label=\textit{(\roman*)}]
\item \longversion{$\mathbf{B},f \models \exists x_0 \ldots \exists x_3\exists w_0 \ldots \exists w_3( \alpha \wedge \psi' )$.}
\shortversion{$\mathbf{B},f \models \exists x_0 \ldots x_3 w_0 \ldots w_3( \alpha \wedge \psi' )$.}
\item $\mathbf{B}^* \models \psi$.
\end{enumerate}
\end{claimm}

\longversion{\pfclnphardd}

We conclude the proof by showing that $\mathbf{B}^* \models \psi$ if and only if $\mathbf{B} \models \phi$.  
If $\mathbf{B} \not\models \phi$, 
then there exists an assignment~$f$ such that 
$\mathbf{B},f \not\models \exists x_0 \ldots \exists x_3\exists w_0 \ldots \exists w_3( \alpha \wedge \psi' )$; 
by Claim~\ref{cl:nphard2}, $f$ is nontrivial.  Then $\mathbf{B}^* \not \models \psi$ by Claim~\ref{cl:nphard4}.
Conversely, if $\mathbf{B} \models \phi$, 
then in particular $\mathbf{B},f \models \exists x_0 \ldots \exists x_3\exists w_0 \ldots \exists w_3( \alpha \wedge \psi' )$ 
for all nontrivial assignments~$f$, and hence
$\mathbf{B}^* \models \psi$ by Claim~\ref{cl:nphard4}.  
\shortversion{\qed}
\end{proof}

\section{Reduced Forms}\label{sect:red}

In this section, we introduce \emph{reduced} forms for conjunctive positive 
sentences on posets and prove that, given a poset $\pp$ and a sentence $\phi$, 
a reduced form for $\phi$ is easy to compute and equivalent to $\phi$ on $\pp$.

\emph{In the rest of this section, $\sigma=\{\leq\}$ is the vocabulary of posets, 
and $\phi$ is a conjunctive positive $\sigma$-sentence as in {\em(\ref{eq:strictform})}.} 
Since $\phi$ will be evaluated on posets, where the formulas $x \leq y \wedge y \leq x$ and $x=y$ are equivalent, 
we assume that no atom of the form $x=y$ occurs in $\phi$; otherwise, 
such an atom can be replaced by the formula $x \leq y \wedge y \leq x$ maintaining logical equivalence. 

We represent $\phi$ by the pair $(\mathbf{Q}_\phi,\mathbf{M}_\phi)$, where $\mathbf{Q}_\phi=(Q_\phi,E^{\mathbf{Q}_\phi})$ and $\mathbf{M}_\phi=(M_\phi,E^{\mathbf{M}_\phi})$ are digraphs 
encoding the \emph{prefix} and the \emph{matrix} of $\phi$ respectively, as follows.  
The universes are $Q_\phi=M_\phi=\{x_1,y_1,\ldots,x_l,y_l\}$; we let $M^\forall_\phi=\{x_1,\ldots,x_l\}$ and $M^\exists_\phi=\{y_1,\ldots,y_l\}$ 
denote, respectively, the set of \emph{universal} and \emph{existential} variables in $\phi$.  The structure 
$\mathbf{Q}_\phi$ is a chain with cover relation 
$x_1 \prec^{\mathbf{Q}_\phi} y_1 \prec^{\mathbf{Q}_\phi} \cdots \prec^{\mathbf{Q}_\phi} x_l \prec^{\mathbf{Q}_\phi} y_l$.  
The structure $\mathbf{M}_\phi$ is defined by the 
edge relation $E^{\mathbf{M}_\phi}=\{(x,y) \mid \text{$x \leq y$ is an atom of $\phi$}\}$.  
We say that $\phi$ is in \emph{reduced form} if: 
\longshort{\begin{enumerate}[label=\textit{(\roman*)},leftmargin=0.8cm]
\item $\mathbf{M}_{\phi}$ is a poset;
\item the substructure of $\mathbf{M}_{\phi}$ induced by $M^{\forall}_{\phi}$ is an antichain; 
\item for all distinct $x$ and $x'$ in $M^\forall_\phi$, it holds that 
$[x)^{\mathbf{M}_{\phi}} \cap [x')^{\mathbf{M}_{\phi}}=(x]^{\mathbf{M}_{\phi}} \cap (x']^{\mathbf{M}_{\phi}}=\emptyset$; 
\item for all $x \in M^\forall_\phi$ and all $y \in M^\exists_\phi \cap ((x]^{\mathbf{M}_{\phi}} \cup [x)^{\mathbf{M}_{\phi}})$, 
it holds that $x<^{\mathbf{Q}_\phi} y$.
\end{enumerate}}
{\begin{enumerate}[label=\textit{(\roman*)},leftmargin=0.8cm]
\item $\mathbf{M}_{\phi}$ is a poset;
\item the substructure of $\mathbf{M}_{\phi}$ induced by $M^{\forall}_{\phi}$ is an antichain; 
\item for all distinct $x$ and $x'$ in $M^\forall_\phi$, it holds that 
$[x)^{\mathbf{M}_{\phi}} \cap [x')^{\mathbf{M}_{\phi}}=(x]^{\mathbf{M}_{\phi}} \cap (x']^{\mathbf{M}_{\phi}}=\emptyset$; 
\item for all $x \in M^\forall_\phi$ and all $y \in M^\exists_\phi \cap ((x]^{\mathbf{M}_{\phi}} \cup [x)^{\mathbf{M}_{\phi}})$, 
it holds that $x<^{\mathbf{Q}_\phi} y$.
\end{enumerate}
}

Let $\phi \in \fo(\forall,\exists,\wedge)$. For all $Z \subseteq M_\phi$, we let $\phi|_Z$ denote the conjunctive positive sentence 
represented by $(\mathbf{Q}_\phi|_Z,\mathbf{M}_\phi|_Z)$.  It is readily observed that, for all $Z \subseteq M_\phi$, 
it holds that $\phi \models \phi|_Z$.


\longshort{\begin{proposition}}{\begin{proposition}[$\star$]}
\label{prop:nf} 
Let $\mathcal{P}$ be a class of posets.  
There exists a polynomial-time algorithm that, 
given an instance $(\pp,\phi)$ of $\textsc{MC}(\mathcal{P},\fo(\forall,\exists,\wedge))$, 
either correctly rejects, 
or returns a sentence $\phi' \in \fo(\forall,\exists,\wedge)$ in reduced form 
such that $\pp \models \phi'$ if and only if $\pp \models \phi$.
\end{proposition}

\newcommand{\pfpropnf}[0]{
\begin{proof} 
The algorithm works as follows.  Let $\phi=(\mathbf{Q}_\phi,\mathbf{M}_\phi)$.  
Let $\phi^*=(\mathbf{Q}_{\phi^*},\mathbf{M}_{\phi^*})$ be such that 
$\mathbf{Q}_{\phi^*}=\mathbf{Q}_\phi$ and 
$\mathbf{M}_{\phi^*}$ is the reflexive transitive closure of $\mathbf{M}_\phi$.  Note that 
$\pp \models \phi$ if and only if $\pp \models \phi^*$, because $\leq^\pp$ is reflexive and transitive.  

The algorithm first fixes clause $(ii)$ in the definition of reduced form.  If $\mathbf{M}_{\phi^*}$ contains a directed edge $(x,x')$ between 
two distinct universal variables $x$ and $x'$, then the algorithm rejects; 
indeed, $\phi^*|_{\{x,x'\}} \equiv \forall x \forall x'(x \leq x')$ and 
$\phi^* \models \forall x \forall x'(x \leq x')$ by the observation before the statement, 
but $\pp \not\models \forall x \forall x'(x \leq x')$ because $\pp$ is nontrivial.  Note that if the algorithm does not terminate at this stage, 
then the substructure of $\mathbf{M}_{\phi^*}$ induced by $M^{\forall}_{\phi^*}$ is an antichain (we use this fact below).

Next, the algorithm fixes clause $(i)$.  As long as $\mathbf{M}_{\phi^*}$ contains directed cycles of length at least $2$, 
the algorithm detects one such cycle and either rejects, 
or reassigns $\mathbf{Q}_{\phi^*}$ and $\mathbf{M}_{\phi^*}$, as follows.  
Let $(z_1,z_2,\ldots,z_m)$ be a directed cycle in $\mathbf{M}_{\phi^*}$ ($m \geq 2$). Since clause $(ii)$ holds and $\phi^*$ is transitive, any such cycle can contain at most one universal variable.  
If the cycle contains exactly one universal variable, say $x$, there are two cases.  
If there exists an existential variable $y$ in the cycle such that $\exists y \forall x$ is a subsequence of $\mathbf{Q}_{\phi^*}$, 
then $\phi^*|_{\{x,y\}} \equiv \exists y \forall x(x=y)$ and $\phi^* \models \exists y \forall x(x=y)$, 
and again the algorithm rejects.  Otherwise, 
let $Z=M_{\phi^*} \setminus (  \{z_1,z_2,\ldots,z_m\} \setminus\{x\})$.  
The algorithm reassigns $\phi^* \leftrightharpoons \phi^*|_{Z}$.   
%
Note that $\pp \models (\mathbf{Q}_{\phi^*},\mathbf{M}_{\phi^*})$ 
if and only if $\pp \models (\mathbf{Q}_{\phi^*}|_{Z},\mathbf{M}_{\phi^*}|_{Z})$, 
using the fact that $\mathbf{M}_{\phi^*}$ is transitively closed; 
in fact the transitive closure warrants that if an atom $w \leq z_i$ was in $\phi^*$, then the atom $w \leq x$ is in $\phi^*|_Z$, 
and if an atom $z_i \leq w$ was in $\phi^*$, then the atom $x \leq w$ is in $\phi^*|_Z$. 
If the cycle contains only existential variables, 
and $y$ is the smallest such variable in $\mathbf{Q}_{\phi^*}$, 
then let $Z=M_{\phi^*} \setminus (  \{z_1,z_2,\ldots,z_m\} \setminus\{y\})$.  
The algorithm reassigns $\phi^* \leftrightharpoons \phi^*|_{Z}$; 
note that $\pp \models (\mathbf{Q}_{\phi^*},\mathbf{M}_{\phi^*})$ if and only if $\pp \models (\mathbf{Q}_{\phi^*}|_{Z},\mathbf{M}_{\phi^*}|_{Z})$.  
At loop termination, $\mathbf{M}_{\phi^*}$ is a poset (we use this fact below).  

Next, the algorithm detects and eliminates violations of clause $(iv)$.  
Let $y \in M^\exists_{\phi^*}$ and $x \in M^\forall_{\phi^*}$ 
be such that $y <^{\mathbf{Q}_{\phi^*}} x$ and (say) $x \leq^{\mathbf{M}_{\phi^*}} y$.  
Therefore, $\phi^*|_{\{y,x\}} \equiv \exists y \forall x(x \leq y)$ 
and $\phi^* \models \phi^*|_{\{y,x\}}$.  If $\pp$ lacks a top element, 
then it is readily checked that $\pp \not\models \phi^*|_{\{y,x\}}$, 
so that $\pp \not\models \phi^*$, and the algorithm rejects.  
Otherwise, if $\pp$ has a top element $t \in P$, 
then let $Z=M_{\phi^*} \setminus [y)^{\mathbf{M}_{\phi^*}}$.  
The algorithm reassigns $\phi^* \leftrightharpoons \phi^*|_{Z}$.  
We check that $\pp \models \phi^*$ if and only if $\pp \models \phi^*|_{Z}$.  
The forward direction holds.  For the backward direction, 
it is readily checked that a winning strategy for Eloise on $\pp$ and $\phi^*|_{Z}$ 
yields a winning strategy for Eloise on $\pp$ and $\phi^*$ 
by sending all variables in $[y)^{\mathbf{M}_{\phi^*}}$ identically to $t$ 
(independent of the play by Abelard).  

Now, the algorithm detects and eliminates violations of clause $(iii)$.  
Let $y \in M^\exists_{\phi^*}$ and $x,x' \in M^\forall_{\phi^*}$, $x \neq x'$, 
be such that (say) $x,x' \leq^{\mathbf{M}_{\phi^*}} y$. Since clause $(iv)$ holds, we have $x,x' <^{\mathbf{Q}_{\phi^*}} y$.
Therefore, $\phi^*|_{\{x,x',y\}} \equiv \forall x \forall x' \exists y(x \leq y \wedge x' \leq y)$ 
and $\phi^* \models \phi^*|_{\{x,x',y\}}$.  As above, if $\pp$ lacks a top element, 
then the algorithm rejects; otherwise, the algorithm reassigns 
$\phi^* \leftrightharpoons \phi^*|_{Z}$, where $Z=M_{\phi^*} \setminus [y)^{\mathbf{M}_{\phi^*}}$.  

Finally, the algorithm assigns $\phi' \leftrightharpoons \phi^*$ 
and returns $\phi'$.  The algorithm runs in polynomial time.  Moreover, 
if it decides the instance, the output is correct; 
and if it does not decide the instance, the returned sentence $\phi'$ is in reduced form, 
and such that $\pp \models \phi'$ if and only if $\pp \models \phi$. 
\shortversion{\qed}\end{proof}}

\longversion{\pfpropnf}

\newcommand{\digression}[0]{
In a slight digression, we observe that Proposition~\ref{prop:nf}, 
in combination with the statement below, allows us to prove that $\textsc{MC}(\{\pp\},\fo(\forall,\exists,\wedge))$ 
is polynomial-time tractable for every poset $\pp$ 
containing an element between all minimal and all maximal elements 
(for instance, posets with a top or bottom, and in particular semilattices).

\longshort{\begin{proposition}}{\begin{proposition}}
\label{prop:magic}
Let $\pp$ be a poset and let $\phi$ be a conjunctive positive sentence in reduced form.  
If there exists $p \in P$ such that $m \leq^\pp p \leq^\pp M$ for all $m \in \mathrm{min}(\pp)$ 
and $M \in \mathrm{max}(\pp)$, then $\pp \models \phi$.
\end{proposition}
}

\longversion{\digression}

\newcommand{\pfpropmagic}[0]{
\begin{proof}
We describe a winning strategy for Eloise; to simplify the notation, 
we assume without loss of generality that $\phi$ is as in (\ref{eq:strictform}), 
hence such a strategy has the form $\mathbf{g}=(g_1,\ldots,g_l)$.  For all $q \in P$, let $m(q)$ denote an arbitrarily fixed element in $(q]^\pp \cap \mathrm{min}(\pp)$, 
and let $M(q)$ denote an arbitrarily fixed element in $[q)^\pp \cap \mathrm{max}(\pp)$.  

Let $p \in P$ be as in the statement of the lemma and let $i \in [l]$.  We distinguish three cases.  If $y_i$ is incomparable in $\mathbf{M}_{\phi}$ to all universal variables, 
then $g_i(q_1,\ldots,q_i)=p$ for all $q_1,\ldots,q_i \in P$.  If $y_i$ is above universal variable $x_j$ in $\mathbf{M}_{\phi}$, 
in symbols $x_j \leq^{\mathbf{M}_{\phi}} y_i$, 
then by clause $(iv)$ in the definition of reduced form it holds that $x_j <^{\mathbf{Q}_{\phi}} y_i$, 
and we let $g_i(q_1,\ldots,q_j,\ldots,q_i)=M(q_j)$ for all $q_1,\ldots,q_i \in P$.  Similarly, 
if $y_i \leq^{\mathbf{M}_{\phi}} x_j$, then $x_j <^{\mathbf{Q}_{\phi}} y_i$, 
and we let $g_i(q_1,\ldots,q_j,\ldots,q_i)=m(q_j)$ for all $q_1,\ldots,q_i \in P$.  
Since $\phi$ satisfies clauses $(i)$, $(ii)$, and $(iii)$ in the definition of reduced form, 
the case distinction is exhaustive, and the definition is sound and complete.  

It is easily checked that $\mathbf{g}$ is a winning strategy for Eloise, 
by a case distinction relying on the fact that $\phi$ is in reduced form.
\shortversion{\qed}\end{proof}}

\longversion{\pfpropmagic}

\longversion{
\begin{corollary}
Let $\mathcal{P}$ be any class of posets with top or bottom (for instance, 
any class of semilattices).  Then, 
$\textsc{MC}(\mathcal{P},\mathcal{FO}(\forall,\exists,\wedge))$ is polynomial-time tractable.
\end{corollary}
\begin{proof}
Let $(\pp,\phi)$ be an instance of $\textsc{MC}(\mathcal{P},\mathcal{FO}(\forall,\exists,\wedge))$.  
The algorithm first invokes the algorithm in Proposition~\ref{prop:nf}, 
which either decides correctly the instance, or returns a sentence $\phi'$ in reduced form such that 
$\pp \models \phi'$ if and only if $\pp \models \phi$; in the latter case, the algorithm accepts.  

The algorithm runs in polynomial time.  For correctness, 
if the algorithm rejects, then it rejects correctly by the correctness of the algorithm in Proposition~\ref{prop:nf}.  
If the algorithm accepts, we claim that $\pp \models \phi$.  Note that 
$\pp \in \mathcal{P}$ implies that $\pp$ has a top or bottom element; 
say that $\pp$ has a top element $t$.  Then, $m \leq^\pp t \leq^\pp M$ 
for all $m \in \mathrm{min}(\pp)$ and $M \in \mathrm{max}(\pp)=\{t\}$, 
and $\pp \models \phi'$ by Proposition~\ref{prop:magic}; 
the claim follows.  
\end{proof}
}


\section{Fixed-Parameter Tractability}\label{sect:cptract}

In this section, we prove that model checking conjunctive positive logic is fixed-parameter tractable 
parameterized by the size of the sentence \emph{and} the width of the poset; it follows, in particular, 
that model checking conjunctive positive logic is fixed-parameter tractable (parameterized by the size of the sentence) 
on classes of posets of bounded width.  We refer the reader to the introduction 
for an informal outline of the proof idea.

\emph{In the rest of this section, $\sigma=\{\leq\}$ is the vocabulary of posets, 
$\pp$ is a poset and $\phi=(\mathbf{Q}_\phi,\mathbf{M}_\phi)$ is a conjunctive 
positive $\sigma$-sentence as in {\em(\ref{eq:strictform})} satisfying clauses $(i)$ and $(ii)$ 
of the definition of reduced form.}  

\longversion{In the sequel we define the two notions of \lq\lq depth\rq\rq\ of a variable in the sentence $\phi$ 
(Section~\ref{subsect:depthsent}) and \lq\lq depth\rq\rq\ of an element in the poset $\pp$ 
(Section~\ref{subsect:depthposet}); we freely override the notation $\mathrm{depth}(\cdot)$, 
already used to measure the depth of a poset.  We then relate the two notions (Section~\ref{subsect:depthrest}), 
from which we obtain the tractability result (Section~\ref{subsect:tract}).}  

\subsection{Depth in the Sentence}\label{subsect:depthsent} 


Using the fact that $\phi$ is in reduced form, we define the following.  For all 
$y \in M^\exists_\phi$: 
\shortversion{$\ldpt(y)=\mathrm{depth}(\mathbf{M}_\phi|_{(y]^{\mathbf{M}_\phi}})$; 
$\udpt(y)=\mathrm{depth}(\mathbf{M}_\phi|_{[y)^{\mathbf{M}_\phi}})$.}
\longversion{\begin{itemize}
\item $\ldpt(y)=\mathrm{depth}(\mathbf{M}_\phi|_{(y]^{\mathbf{M}_\phi}})$;
\item $\udpt(y)=\mathrm{depth}(\mathbf{M}_\phi|_{[y)^{\mathbf{M}_\phi}})$.
\end{itemize}}
In words, 
$\ldpt(y)$ is the size of the largest chain in the substructure of 
$\mathbf{M}_\phi$ induced by the downset of $y$ in $\mathbf{M}_\phi$, 
and $\udpt(y)$ is the size of the largest chain in the substructure of 
$\mathbf{M}_\phi$ induced by the upset of $y$ in $\mathbf{M}_\phi$.

Next, we define a partition of $M^\exists_\phi$ into two blocks $L_\phi$ and $U_\phi$, 
the \emph{lower} and \emph{upper} variables respectively,  
as follows.  For all $y \in M^\exists_\phi$ let 
\shortversion{$y \in L_\phi$ if and only if there either exists $x \in M^\forall_\phi$ such that $y \leq^{\mathbf{M}_\phi} x$, 
or $y \parallel^{\mathbf{M}_\phi} x$ for all $x \in M^\forall_\phi$ and 
$\ldpt(y) \leq \udpt(y)$. Similarly, 
$y \in U_\phi$ if and only if there either exists $x \in M^\forall_\phi$ such that $y \geq^{\mathbf{M}_\phi} x$, 
or $y \parallel^{\mathbf{M}_\phi} x$ for all $x \in M^\forall_\phi$ and 
$\ldpt(y)>\udpt(y)$.}
\longversion{\begin{itemize}
\item $y \in L_\phi$ if and only if there either exists $x \in M^\forall_\phi$ such that $y \leq^{\mathbf{M}_\phi} x$, 
or $y \parallel^{\mathbf{M}_\phi} x$ for all $x \in M^\forall_\phi$ and 
$\ldpt(y) \leq \udpt(y)$; 
\item $y \in U_\phi$ if and only if there either exists $x \in M^\forall_\phi$ such that $y \geq^{\mathbf{M}_\phi} x$, 
or $y \parallel^{\mathbf{M}_\phi} x$ for all $x \in M^\forall_\phi$ and 
$\ldpt(y)>\udpt(y)$.
\end{itemize}}
In words, an existential variable $y$ in $\phi$ is lower if and only if 
it is below a universal variable in the matrix of $\phi$, 
or is incomparable to all universal variables in the matrix of $\phi$ 
but \lq\lq closer\rq\rq\ to the bottom of the matrix of $\phi$ in that $\ldpt(y) \leq \udpt(y)$; 
a similar idea drives the definition of upper variables.

Finally we define, for all $y \in M^\exists_\phi$: 
\shortversion{$\mathrm{depth}(y) = \ldpt(y)$ if $y \in L_\phi$, 
and $\mathrm{depth}(y) = \udpt(y)$ if $y \in U_\phi$;}
\longversion{\begin{align*}
\mathrm{depth}(y) &= 
\begin{cases}
\ldpt(y)\text{,} & \text{if $y \in L_\phi$,}\\
\udpt(y)\text{,} & \text{if $y \in U_\phi$;}
\end{cases}
\end{align*}}in words, the depth of a lower variable is its \lq\lq distance\rq\rq\ from the bottom as measured by $\ldpt(y)$, and similarly for upper variables.

\subsection{Depth in the Structure}\label{subsect:depthposet} 


Relative to the poset $\pp$, we define, for all $i \geq 0$, the set $P_i$ as follows.
\begin{itemize}
\item $L_0=\min(\pp)$, $U_0=\max(\pp) \setminus L_0$, and $P_0=L_0 \cup U_0$.  
\item Let $i \geq 1$,  
and let $R \subseteq P_{i-1}$ be such that $R \cap L_{i-1}$ is downward closed in $\pp|_{L_{i-1}}$ 
(that is, for all $l,l' \in L_{i-1}$, if $l \in R \cap L_{i-1}$ and $l' \leq^\pp l$, then $l' \in R$) 
and $R \cap U_{i-1}$ is upward closed in $\pp|_{U_{i-1}}$ 
(that is, for all $u,u' \in U_{i-1}$, if $u \in R \cap U_{i-1}$ and $u \leq^\pp u'$, then $u' \in R$).  
Let
$$P_{i-1,R}=\left\{ p \in P ~\left|~ \begin{array}{l}
\textup{for all $l \in L_{i-1}$, $l \leq^{\mathbf{P}} p$  if and only if $l \in R$,} \\ 
\textup{for all $u \in U_{i-1}$, $p \leq^{\mathbf{P}} u$  if and only if $u \in R$}
\end{array}\right. \right\}\text{;}$$
in words, $p \in P_{i-1,R}$ if and only if 
the elements in $L_{i-1}$ below $p$ are exactly those in $R \cap L_{i-1}$ 
(and the elements in $L_{i-1} \setminus R$ are incomparable to $p$) 
and the elements in $U_{i-1}$ above $p$ are exactly those in $R \cap U_{i-1}$ 
(and the elements in $U_{i-1} \setminus R$ are incomparable to $p$). 
We now define $P_i = L_i \cup U_i$ where $L_i$ and $U_i$ are as follows:
\shortversion{$$L_{i} = L_{i-1} \cup \bigcup_{R \subseteq P_{i-1}} \mathrm{min}^{\mathbf{P}}( P_{i-1,R})\text{, }
U_{i} = \big( U_{i-1} \cup \bigcup_{R \subseteq P_{i-1}} \mathrm{max}^{\mathbf{P}}( P_{i-1,R} )\big)\setminus L_i\text{.}$$}
\longversion{\begin{align*}
L_{i} &= L_{i-1} \cup \bigcup_{R \subseteq P_{i-1}} \mathrm{min}^{\mathbf{P}}( P_{i-1,R})\text{,}\\
U_{i} &= ( U_{i-1} \cup \bigcup_{R \subseteq P_{i-1}} \mathrm{max}^{\mathbf{P}}( P_{i-1,R} ))\setminus L_i\text{.}
\end{align*}}
\end{itemize}

Let $p \in P$.  Let $i \geq 0$ be minimum such that $p \in P_i$ 
(note that for every $p \in P$ such minimum $i$ exists, 
and $L_i \cap U_i=\emptyset$ by construction).  
\longshort{Then: 
\begin{itemize}
\item if $p \in L_i$, then $p \in L_{\mathbf{P}}$ and $\ldpt(p)=i$;
\item if $p \in U_i$, then $p \in U_{\mathbf{P}}$ and $\udpt(p)=i$.
\end{itemize}}{If $p \in L_i$, then $p \in L_{\mathbf{P}}$ and $\ldpt(p)=i$, and 
if $p \in U_i$, then $p \in U_{\mathbf{P}}$ and $\udpt(p)=i$.}
Note that $L_{\mathbf{P}}$ and $U_{\mathbf{P}}$ partition $P$ into two blocks containing the \emph{lower} and \emph{upper} elements respectively. 
Finally we define, for all $p \in P$: 
\shortversion{$\mathrm{depth}(p) = \ldpt(p)$, if $p \in L_{\mathbf{P}}$, 
and $\mathrm{depth}(p) = \udpt(p)$, if $p \in U_{\mathbf{P}}$.}
\longversion{\begin{align*}
\mathrm{depth}(p) &= 
\begin{cases}
\ldpt(p)\text{,} & \text{if $p \in L_{\mathbf{P}}$,}\\
\udpt(p)\text{,} & \text{if $p \in U_{\mathbf{P}}$.}
\end{cases}
\end{align*}}

\subsection{Depth Restricted Game}\label{subsect:depthrest}  

We now establish and formalize the relation between the depth in $\phi$ and the depth in $\pp$ (see Lemma~\ref{lemma:eloiserestr});
this is the key combinatorial fact underlying the model checking algorithm.

Relative to the Hintikka game on $\pp$ and $\phi$, we define the following.  
A pair $(y,p) \in M^\exists_\phi \times P$ is \emph{depth respecting} 
\shortversion{if $(y,p) \in (L_{\phi} \times L_{\mathbf{P}}) \cup (U_{\phi} \times U_{\mathbf{P}})$
and $\dpt(p) \leq \dpt(y)$.}  
\longversion{if $$(y,p) \in (L_{\phi} \times L_{\mathbf{P}}) \cup (U_{\phi} \times U_{\mathbf{P}})$$
and $$\dpt(p) \leq \dpt(y)\text{.}$$}  
A strategy $(g_1,\ldots,g_l)$ for Eloise 
is \emph{depth respecting} if, 
for all $i \in [l]$ and all plays $f \colon \{x_1,\ldots,x_l\} \to P$ by Abelard, 
the pair $(y_i,g_i(f(x_1),\ldots,f(x_i)))$ is depth respecting.

Let $b \geq 0$ be the maximum depth of a variable in $\phi$.  
A play $f \colon \{x_1,\ldots,x_l\} \to P$ by Abelard is \emph{bounded depth} 
if, for all $i \in [l]$, it holds that $f(x_i) \in P_{b+1}$.  

\begin{lemma}
\label{lemma:eloiserestr}
The following are equivalent (w.r.t.\   the Hintikka game on $\pp$ and $\phi$).
\begin{enumerate}[label=\textit{(\roman*)}]
\item Eloise has a winning strategy.
\item Eloise has a depth respecting 
winning strategy.
\item Eloise has a depth respecting strategy beating all bounded depth Abelard plays.
\end{enumerate}
\end{lemma}

\newcommand{\pfcle}[0]{
\begin{proof}
Recall that $\phi$ is in reduced form, 
hence it does not contain atoms of the form $x \leq x'$ with $x \neq x'$ and $x,x' \in M_\phi^\forall$. 

Assume $(f'(x_1),\ldots,f'(x_j)) \neq (f(x_1),\ldots,f(x_j))$.  
Since $f'$ and $\mathbf{g}$ satisfy all atoms in $\phi$ by hypothesis, 
and the assignment of $y_j$ in $P$ induced by $f'$ and $\mathbf{g}'$ is equal to the assignment of $y_j$ in $P$ induced by $f'$ and $\mathbf{g}$, 
it follows that $f'$ and $\mathbf{g}'$ satisfy all atoms in $\phi$.

Assume that $(f'(x_1),\ldots,f'(x_j))=(f(x_1),\ldots,f(x_j))$.  
Observe that $f'$ and $\mathbf{g}'$ satisfy all atoms of the form $z \leq z'$ where $y_j$ does not occur, 
because $f'$ and $\mathbf{g}$ satisfy all such atoms by hypothesis, 
and the assignment of any variable distinct from $y_j$ in $P$ 
induced by $f'$ and $\mathbf{g}$ is equal to the assignment of any such variable in $P$ 
induced by $f'$ and $\mathbf{g}'$ by construction.  

Thus, suffices to check atoms in $\phi$ where $y_j$ occurs.  
Consider an atom of the form $y_j \leq z$.  If $z=y_j$ the atom is trivially satisfied.  
If $z \neq y_j$, as observed above the assignment of $z$ in $P$ is unchanged in passing from 
$f'$ and $\mathbf{g}$ to $f'$ and $\mathbf{g}'$; hence, 
the atom is satisfied under the assignment in $P$ induced by $f'$ and $\mathbf{g}'$, 
because $m \leq^\mathbf{P} p$.

Now consider an atom of the form $z \leq y_j$.  If $z=y_j$ the atom is trivially satisfied.  
Assume $z \neq y_j$.  Since $y_j \in L_\phi$ in the case under analysis, 
it holds that $z \in M^\exists_\phi$ by construction; say $z=y_{j'}$, $j' \in [l]$, $j' \neq j$, 
so that the atom under consideration is $y_{j'} \leq y_j$.  Let $p'=g_{j'}(f'(x_1),\ldots,f'(x_{j'}))$.  
Since the original strategy $\mathbf{g}$ beats $f'$, we have $p' \leq^{\pp} p$; 
we want to show that $p' \leq^{\pp} m$. 

Since the atom $y_{j'} \leq y_j$ is in $\phi$ and $\phi$ is in reduced form, by construction $\dpt(y_{j'})<\dpt(y_j)$, 
so that $y_{j'} \in L_\phi$.  By the choice of the minimal witnesses $f$ and $y_j$, 
it holds that $(y_{j'},g_{j'}(f'(x_1),\ldots,f'(x_{j'})))$ is depth respecting, 
that is, $p'=g_{j'}(f'(x_1),\ldots,f'(x_{j'})) \in L_{\mathbf{P}}$ 
and $\dpt(p') \leq \dpt(y_{j'}) \leq i-1$.  Hence, $p' \in L_{i-1}$ by construction, 
and since $p' \leq^{\pp} p$, we also have that $p' \in R$, where $R$ is the subset of $P_{i-1}$ defined above.  
Since also $m \in P_{i-1,R}$, we have that $l \leq^{\mathbf{P}} m$ for all $l \in R \cap L_{i-1}$; 
in particular, $p' \leq^{\mathbf{P}} m$, and we are done.
\shortversion{\qed}\end{proof}
}

\newcommand{\pfclf}[0]{
\begin{proof}
Let 
$f \colon \{x_1,\ldots,x_l\} \to P$ be any play by Abelard.  We have that 
the assignment in $P$ induced by $f'$ (defined as above relative to $f$) and $\mathbf{g}$ satisfies all atoms in $\phi$; 
we want to show that the assignment in $P$ induced by $f$ and~$\mathbf{g}'$ (defined as above relative to $f'$ and $\mathbf{g}$) satisfies all atoms in $\phi$. 
We enter a case distinction.  Note that, since the substructure of $\mathbf{M}_\phi$ induced by $M^\forall_\phi$ is an antichain, 
there are no atoms of the form $x \leq x'$ with $x \neq x'$ and $x,x' \in M^\forall_\phi$.

All atoms of the form $y \leq y'$ where $y,y' \in M^\exists_\phi$ are satisfied, 
because for all variables in $M^\exists_\phi$, their assignment in $P$ 
induced by $f$ and $\mathbf{g}'$ is equal to their assignment in $P$ 
induced by $f'$ and $\mathbf{g}$ (and the latter is satisfying by hypothesis).  

We conclude considering atoms of the form $x \leq y$ or $y \leq x$, 
where $x \in M^\forall_\phi$ and $y \in M^\exists_\phi$.  Say $x=x_i$ and $y=y_{j}$ for $i,j \in [l]$.  
Consider any atom $y_{j} \leq x_{i}$; the argument is symmetric for any atom $x_i \leq y_{j}$.  
We have that $y_{j} \in L_\phi$, 
and since $\mathbf{g}$ is depth respecting, 
$g_{j}(f'(x_{1}),\ldots,f'(x_{j})) \in L_{\dpt(y_{j})} \subseteq L_{b}$.  
By $(iii)$, we have $g_{j}(f'(x_{1}),\ldots,f'(x_{j})) \leq^{\pp} f'(x_{i})$, 
and $g_{j}(f'(x_{1}),\ldots,f'(x_{j}))=g'_{j}(f(x_1),\ldots,f(x_{j}))$ by definition.  
Then by construction we have $f'(x_i)=r_i$ such that, for all $l \in L_{b}$, it holds 
$l \leq^{\pp} r_i$ if and only if $l \leq^{\pp} p_i=f(x_i)$, and we are done. 
\shortversion{\qed}
\end{proof}}

\begin{proof}
$(ii) \Rightarrow (iii)$ is trivial.  We prove $(i) \Rightarrow (ii)$ and $(iii) \Rightarrow (i)$.  

$(i) \Rightarrow (ii)$: Let $\mathbf{g}=(g_1,\ldots,g_l)$ be a winning strategy for Eloise.  
Let the Abelard play $f \colon \{x_1,\ldots,x_l\} \to P$ and the existential variable $y_j \in M^\exists_\phi$ 
be a minimal witness that the above winning strategy for Eloise is not depth respecting, in the following sense: 
\shortversion{$(y_j,g_j(f(x_1),\ldots,f(x_j)))$ is not depth respecting, but 
for all $f' \colon \{x_1,\ldots,x_l\} \to P$ 
and all $y_{j'} \in M^\exists_\phi$ such that 
either $y_j,y_{j'} \in L_\phi$ and $\ldpt(y_{j'})<\ldpt(y_j)$, 
or $y_j,y_{j'} \in U_\phi$ and $\udpt(y_{j'})<\udpt(y_j)$, 
it holds that $(y_{j'},g_j(f'(x_1),\ldots,f'(x_j)))$ is depth respecting.}
\longversion{\begin{itemize}
\item $(y_j,g_j(f(x_1),\ldots,f(x_j)))$ is not depth respecting;
\item for all $f' \colon \{x_1,\ldots,x_l\} \to P$ 
and all $y_{j'} \in M^\exists_\phi$ such that 
either $y_j,y_{j'} \in L_\phi$ and $\ldpt(y_{j'})<\ldpt(y_j)$, 
or $y_j,y_{j'} \in U_\phi$ and $\udpt(y_{j'})<\udpt(y_j)$, 
it holds that $(y_{j'},g_j(f'(x_1),\ldots,f'(x_j)))$ is depth respecting.
\end{itemize}}

We define a strategy $\mathbf{g}'=(g_1,\ldots,g_{j-1},g'_j,g_{j+1},\ldots,g_l)$ for Eloise 
such that $g'_j$ restricted to $P^j \setminus \{(f(x_1),\ldots,f(x_j))\}$ is equal to $g_j$ 
(in other words, $g'_j$ differs from $g_j$ only in the move after $f \colon \{x_1,\ldots,x_l\} \to P$), 
and $(y_j,g'_j(f(x_1),\ldots,f(x_j)))$ is depth respecting.  
There are two cases to consider, depending on whether $y_j \in L_\phi$ or $y_j \in U_\phi$.  
We prove the statement in the former case; the argument is symmetric in the latter case.  

So, assume $y_j \in L_\phi$.  Let $g_j(f(x_1),\ldots,f(x_j))=p$ and $i=\dpt(y_j)$.    Let $R \subseteq P_{i-1}$ (with $R \cap L_{i-1}$ downward closed in $\pp|_{L_{i-1}}$ and 
$R \cap U_{i-1}$ upward closed in $\pp|_{U_{i-1}}$) 
be such that, for all $l \in L_{i-1}$ and $u \in U_{i-1}$, it holds that 
$l \leq^{\mathbf{P}} p$ if and only if $l \in R$ and $p \leq^{\mathbf{P}} u$ if and only if $u \in R$.  Hence $p \in P_{i-1,R}$.  
Then there exists $m \in \mathrm{min}^{\mathbf{P}}( P_{i-1,R})$ 
such that $m \leq^{\mathbf{P}} p$. By construction we have $\dpt(m)=i$.  Let $g'_j \colon P^j \to P$ 
be exactly as $g_j$ with the exception that $g'_j(f(x_1),\ldots,f(x_j))=m$; 
note that the pair $(y_j,m)$ is depth respecting.  

\longshort{\begin{claimm}}{\begin{claimm}[$\star$]}
\label{cl:erestra}
Let $f'$ be any play by Abelard. 
Then $\mathbf{g}'=(g_1,\ldots,g'_j,\ldots,g_l)$ beats $f'$ in the Hintikka game on $\pp$ and $\phi$.
\end{claimm}

\longversion{\pfcle}

We obtain a depth respecting winning strategy for Eloise by iterating the above argument thanks to Claim 5. 


$(iii) \Rightarrow (i)$: Let $b \geq 0$ be the maximum depth of a variable in $\phi$, 
and let $\mathbf{g}=(g_1,\ldots,g_l)$ be a depth respecting strategy for Eloise 
beating all bounded depth plays by Abelard.  We define a strategy $\mathbf{g}'=(g'_1,\ldots,g'_l)$ for Eloise, as follows. 

Let $f \colon \{x_1,\ldots,x_l\} \to P$ be a play by Abelard, say $f(x_i)=p_i$ for all $i \in [l]$. Let $i \in [l]$ and let $R_i \subseteq P_{b}$ 
(with $R_i \cap L_{b-1}$ downward closed in $\pp|_{L_{b-1}}$ and 
$R_i \cap U_{b-1}$ upward closed in $\pp|_{U_{b-1}}$)
be such that 
for all $l \in L_{b}$, it holds that $l \leq^{\mathbf{P}} p_i$ if and only if $l \in R_i$ 
and for all $u \in U_{b}$, it holds that $p_i \leq^{\mathbf{P}} u$ if and only if $u \in R_i$.  By construction, 
there exists $r_i \in P_{b+1}$ such that for all $l \in L_{b}$, it holds that $l \leq^{\mathbf{P}} r_i$ if and only if 
$l \leq^{\mathbf{P}} p_i$ and for all $u \in U_{b}$, it holds that $r_i \leq^{\mathbf{P}} u$ if and only if $p_i \leq^{\mathbf{P}} u$.
Let $f' \colon \{x_1,\ldots,x_l\} \to P$ be the bounded depth play by Abelard 
\shortversion{defined by $f'(x_i)=r_i$ 
for all $i \in [l]$.  Finally define, for all $i \in [l]$, 
$g'_i(f(x_1),\ldots,f(x_i))=g_i(f'(x_1),\ldots,f'(x_i))$.}
\longversion{defined by $$f'(x_i)=r_i$$
for all $i \in [l]$.  Finally define, for all $i \in [l]$,
$$g'_i(f(x_1),\ldots,f(x_i))=g_i(f'(x_1),\ldots,f'(x_i))\text{.}$$}

\longshort{\begin{claimm}}{\begin{claimm}[$\star$]}
\label{cl:erestrb}
$\mathbf{g}'=(g'_1,\ldots,g'_l)$ is a winning strategy for Eloise. 
\end{claimm}
\longversion{\pfclf}
This concludes the proof of the lemma. \shortversion{\qed}
\end{proof}

\subsection{Fixed-Parameter Tractability}\label{subsect:tract} 

The following two lemmas allow to establish 
the correctness (Lemma~\ref{lemma:relativization}, relying on Lemma~\ref{lemma:eloiserestr}) 
and the tractability (Lemma~\ref{lemma:boundedsearch}) of the presented model checking algorithm, 
respectively.

\longshort{\begin{lemma}}{\begin{lemma}[$\star$]}
\label{lemma:relativization}
Let $b \geq 0$ be the maximum depth of a variable in $\phi$.  
Let $D=P_{b+1}$ and, for all $i \in [l]$, let 
\begin{align*}
D_i &= 
\begin{cases}
L_{\dpt(y_i)}\text{,} & \text{if $y_i \in L_\phi$,}\\
U_{\dpt(y_i)}\text{,} & \text{if $y_i \in U_\phi$.}
\end{cases}
\end{align*}
Then, 
\shortversion{$\pp \models \phi$ if and only if $\pp \models (\forall x_1 \in D)(\exists y_1 \in D_1)\ldots(\forall x_l \in D)(\exists y_l \in D_l)C$.}
\longversion{$\pp \models \phi$ if and only if $$\pp \models (\forall x_1 \in D)(\exists y_1 \in D_1)\ldots(\forall x_l \in D)(\exists y_l \in D_l)C(x_1,y_1,\ldots,x_l,y_l)\text{.}$$}
\end{lemma}

\newcommand{\pflemrela}[0]{
\begin{proof}
We know that $\pp \models \phi$ if and only if Eloise has a winning strategy in the Hintikka game on $\pp$ and $\phi$, 
as per Item $(i)$ in Lemma~\ref{lemma:eloiserestr}.

We also observed in Section~\ref{sect:prelim} that $\pp \models (\forall x_1 \in D)(\exists y_1 \in D_1)\ldots(\forall x_l \in D)(\exists y_l \in D_l)C(x_1,y_1,\ldots,x_l,y_l)$ 
if and only if, in the Hintikka game on $\rela$ and $\phi$, 
Eloise has a strategy of the form $g_i \colon D^i \to D_i$ for all $i \in [l]$, 
beating all plays $f$ by Abelard such that $f(x_i) \in D$ for all $i \in [l]$; 
in other words, if and only if, in the Hintikka game on $\rela$ and $\phi$, 
Eloise has a depth respecting strategy beating all bounded depth plays by Abelard, 
as per Item $(iii)$ in Lemma~\ref{lemma:eloiserestr}.

Since Item $(i)$ and Item $(iii)$ are equivalent by Lemma~\ref{lemma:eloiserestr}, 
the statement follows.
\shortversion{\qed}\end{proof}}

\longversion{\pflemrela}


\longshort{\begin{lemma}}{\begin{lemma}[$\star$]}
\label{lemma:boundedsearch}
Let $w=\mathrm{width}(\pp)$ and let $k \geq 0$.  
Then, $|P_{k}| \leq 2w^{(3w)^k}$.
\end{lemma}

\newcommand{\pflemboundedsearch}[0]{
\begin{proof}
Induction on $k \geq 0$.  If $k=0$, then $P_{0}=\min(\mathbf{P}) \cup \max(\mathbf{P})$, 
hence $|P_0| \leq 2w$. Let $k \geq 0$, and assume inductively that $|P_{k}| \leq 2w^{(3w)^k}$. 
Since there is a bijective correspondence between downward closed sets in $L_k$ 
and antichains in (the substructure of $\mathbf{M}_\phi$) induced by $L_k$, 
the number of downward closed sets in $L_k$ is bounded above by $|P_k|^w$. 
Similarly, the number of upward closed sets in $U_k$ is bounded above by $|P_k|^w$. 
Then the number of admissible choices for $R\subseteq P_k$ is bounded above by $|P_k|^{2w}$, 
since $R$ is upward closed in $U_k$ and downward closed in $L_k$. 
For any such fixed $R$, we have $|\mathrm{min}^{\mathbf{P}}( P_{i-1,R})|\leq w$ and $|\mathrm{max}^{\mathbf{P}}( P_{i-1,R})|\leq w$. 
Then $|P_{k+1}|\leq |P_k|+|P_k|^{2w}\cdot 2w = 2w^{(3w)^k}+(2w^{(3w)^k})^{2w}\cdot 2w \leq 2w^{(3w)^{k+1}}$.
\shortversion{\qed}\end{proof}}

\longversion{\pflemboundedsearch}

We are now ready to describe the announced algorithm.  The underlying idea is 
that the characterization in Lemma~\ref{lemma:relativization} is checkable in fixed-parameter tractable 
time since $|D_i| \leq |D|$ for all $i \in [l]$, 
and $|D|$ is bounded above by a computable function of $\mathrm{width}(\pp)$ and $\|\phi\|$.  

\longshort{\begin{theorem}}{\begin{theorem}[$\star$]}
\label{th:tractparamwidthsentence}
There exists an algorithm that, 
given a poset $\pp$ and a sentence $\phi \in \mathcal{FO}(\forall,\exists,\wedge)$, 
decides whether $\pp \models \phi$ in 
$$\mathrm{exp}^4_w(O(k)) \cdot n^{O(1)}$$ 
time, 
where $w=\mathrm{width}(\pp)$, $k=\|\phi\|$, and $n=\|(\pp,\phi)\|$.
%

\end{theorem}

\newcommand{\pfthtractpa}[0]{
\begin{proof}
Let $\mathcal{P}$ be any class of posets, 
and let $(\pp,\phi)$ be an instance of $\textsc{MC}(\mathcal{P},\mathcal{FO}(\forall,\exists,\wedge))$.  
Let $w=\mathrm{width}(\pp)$, $k=\|\phi\|$, and $n=\|(\pp,\phi)\|$.  

The algorithm first invokes the algorithm described in Proposition~\ref{prop:nf}, 
which either correctly decides the instance 
or returns a sentence $\phi'$ in reduced form such that $\pp \models \phi'$ 
if and only if $\pp \models \phi$; this is feasible in time $n^{O(1)}$.  

In the latter case, the algorithm constructs $D_i$ for all $i \in [l]$ and $D$ as in Lemma~\ref{lemma:relativization}; 
this is feasible in time $(l+1)|D| \cdot n^{O(1)} \leq k|D| \cdot n^{O(1)}$. 

Next, the algorithm builds all depth respecting strategies for Eloise in the Hintikka game on $\pp$ and $\phi$ 
and for each such strategy checks whether it beats all bounded depth plays by Abelard.  
Note that $D_i \subseteq D_j$ for all $i<j$ in $[l]$ 
and $D_i \subseteq D$ for all $i \in [l]$, hence there are at most $|D|^{|D|^l}$ depth respecting strategies.  
Moreover, there are $|D|^l$ bounded depth plays, 
and checking whether a strategy beats a play is feasible in $n^{O(1)}$ time; 
thus, this step is feasible in time $|D|^{|D|^l} \cdot |D|^l \cdot n^{O(1)}$.  

Combining the above steps, the total runtime is bounded above in $|D|^{|D|^{O(k)}} \cdot n^{O(1)}$.  
By Lemma~\ref{lemma:boundedsearch}, $|D|\leq 2w^{(3w)^{k+1}}$.  
Thus, the runtime is bounded above in $\mathrm{exp}^4_w(O(k)) \cdot n^{O(1)}$, 
where $\mathrm{exp}_{w}^{i+1}(x)=\mathrm{exp}_w(\mathrm{exp}_{w}^{i}(x))=w^{\mathrm{exp}_{w}^{i}(x)}$.  
The algorithm is correct by Lemma~\ref{lemma:relativization}, 
and the statement is proved.\shortversion{\qed}\end{proof}}

\longversion{\pfthtractpa}

\begin{corollary}\label{cor:bwtract}
Let $\mathcal{P}$ be a class of posets of bounded width.  Then, the problem 
$\textsc{MC}(\mathcal{P},\mathcal{FO}(\forall,\exists,\wedge))$ is fixed-parameter tractable.
\end{corollary}

\section{Fixed-Parameter Intractability}\label{sect:hardness}

%

\shortversion{In this section, we prove that model checking 
conjunctive positive logic on classes of bounded depth and bounded cover-degree posets 
is $\textup{coW}[2]$-hard, and hence unlikely to be fixed-parameter tractable \cite{FlumGrohe06}.} 
\longversion{In this section, we prove that there exist classes of posets of bounded 
depth and classes of posets of bounded cover-degree where model checking conjunctive positive 
logic is $\textup{coW}[2]$-hard; thus the problem is unlikely to be fixed parameter tractable, 
since if any $\textup{coW}[i]$-hard problem is fixed-parameter tractable ($i \geq 1$), 
then $\textup{coW}[i] = \textup{FPT} = \textup{coFPT} = \textup{W}[i]$ follows, 
which causes the Exponential Time Hypothesis to fail \cite{FlumGrohe06}.}  

We first observe the following.  Let 
$\phi_k$ be the $\mathcal{FO}(\forall,\exists,\wedge)$-sentence ($k \geq 1$) 
\begin{equation}\label{eq:posetcompl}
\forall x_1 \ldots \forall x_k \exists y_1 \ldots \exists y_k \exists w\left( \bigwedge_{i \in [k]} y_i \leq x_i \wedge \bigwedge_{i \in [k]} y_i \leq w  \right)\text{.} 
\end{equation}

\longshort{\begin{proposition}}{\begin{proposition}[$\star$]}\label{prop:cpprop}
For every poset $\mathbf{P}$ and $k \geq 1$, 
$\mathbf{P} \models \phi_k$ 
iff  
for every $k$ elements $p_1,\ldots,p_k \in \min(\mathbf{P})$, 
there exists $u \in P$ such that $p_1,\ldots,p_k \leq^\mathbf{P} u$.
%
\end{proposition}

\newcommand{\pfpropcpprop}[0]{\begin{proof}
Let $\mathbf{P}$ be a poset and let $k \geq 1$.  If $\mathbf{P} \models \phi_k$, 
then every $k$ elements $r_1,\ldots,r_k \in P$ 
have $k$ lower bounds $p_1,\ldots,p_k \in P$, 
without loss of generality minimal in $\mathbf{P}$, 
having a common upper bound $u \in P$.  Conversely, 
if $r_1,\ldots,r_k$ are any $k$ elements in $\mathbf{P}$, 
then let $p_1,\ldots,p_k$ be $k$ minimal elements in $\mathbf{P}$ 
such that $p_i \leq^{\mathbf{P}} r_i$ for all $i \in [k]$.  
By hypothesis, there exists $u \in P$ 
such that $p_1,\ldots,p_k \leq^{\mathbf{P}} u$, 
hence $\mathbf{P} \models \phi_k$.  
\end{proof}}

\longversion{\pfpropcpprop}

\newcommand{\constrcoverdegree}[0]{
Let $\mathbf{H}$ 
be a hypergraph.  Then, $c(\mathbf{H})=\mathbf{P}$ is the poset defined as follows.  
\begin{itemize}
\item The set of minimal (respectively, maximal), 
elements in $\mathbf{P}$ is $\min(\mathbf{P})=H$ (respectively, $\max(\mathbf{P})=\sigma$).  
\item Let $h \in \min(\mathbf{P})$, let $\mathrm{notin}(h)=\{ U \in \sigma \mid h \not\in U^\mathbf{H} \}$, 
and let $\mathbf{T}_h=(T_h,E^{\mathbf{T}_h})$ be a binary tree rooted at $h$, oriented away from $h$, 
whose outdegree zero nodes correspond exactly to the elements of $\mathrm{notin}(h)$.  Then, 
$T_h \subseteq P$ and the cover relation of $\mathbf{P}$, restricted to $T_h$, is equal to the edge relation of $\mathbf{T}_h$.  
Here, we assume that $T_h \cap T_{h'}=\emptyset$ if $h,h' \in \min(\mathbf{P})$, $h \neq h'$.
\item Let $U \in \max(\mathbf{P})$, let $\mathrm{notin}(U)=\{ h \in H \mid h \not\in U^\mathbf{H} \}$, 
and let $\mathbf{T}_U=(T_U,E^{\mathbf{T}_U})$ be a binary tree rooted at $U$, oriented towards $U$, 
whose indegree zero nodes correspond exactly to the elements of $\mathrm{notin}(U)$.  Then, 
$T_U \subseteq P$ and the cover relation of $\mathbf{P}$, restricted to $T_U$, is equal to the edge relation of $\mathbf{T}_U$.  
Here, we assume that $T_U \cap T_{U'}=\emptyset$ if $U,U' \in \max(\mathbf{P})$, $U \neq U'$.
\item For all $h \in \min(\mathbf{P})$ and $U \in \max(\mathbf{P})$, 
if $h \not\in U^\mathbf{H}$, 
$l$ is the outdegree zero node of $\mathbf{T}_h$ corresponding to $U \in \mathrm{notin}(h)$, 
and $l'$ is the indegree zero node of $\mathbf{T}_U$ corresponding to $h \in \mathrm{notin}(U)$, 
then $l \prec^\mathbf{P} l'$.  
\end{itemize}}

We now describe the reductions.  Let $\mathcal{H}$ be the class of hypergraphs 
(a \emph{hypergraph} is a $\sigma$-structure $\mathbf{H}$ 
such that $U^\mathbf{H} \neq \emptyset$ for all $U$ in a unary vocabulary $\sigma$).  
For the depth invariant, we define a function $d$ from $\mathcal{H}$ 
to a class of posets of depth at most $2$ where $d(\mathbf{H})=\mathbf{P}$ such that: 
\shortversion{$\min(\mathbf{P})=H$; $\max(\mathbf{P})=\sigma$; $h \prec^\mathbf{P} U$ for all $h \in \min(\mathbf{P})$ and $U \in \max(\mathbf{P})$ 
such that $h \not\in U^\mathbf{H}$.}
\longversion{\begin{itemize}
\item $\min(\mathbf{P})=H$;
\item $\max(\mathbf{P})=\sigma$; 
\item $h \prec^\mathbf{P} U$ for all $h \in \min(\mathbf{P})$ and $U \in \max(\mathbf{P})$ 
such that $h \not\in U^\mathbf{H}$.  
\end{itemize}}
For the cover-degree invariant, 
we similarly define a function $c$ from $\mathcal{H}$ to a class of posets with cover graphs of degree at most $3$\longversion{, as follows. \constrcoverdegree}
\shortversion{(see \cite{BovaGanianSzeiderIPEC14} for details). We then use Proposition~\ref{prop:cpprop} to obtain:}  

\longshort{\begin{proposition}}{\begin{proposition}[$\star$]}\label{prop:allhardnessresults} 
Let $r \in \{c,d\}$. Then, $\textsc{MC}(\{ r(\mathbf{H}) \mid \mathbf{H} \in \mathcal{H} \},\mathcal{FO}(\forall,\exists,\wedge))$ is $\textup{coW}[2]$-hard.\end{proposition}

\newcommand{\pfpropdepthh}[0]{
\begin{proof}
Case $r=d$.  We give a fpt many-one reduction from the complement of $\textsc{HittingSet}$ 
to $\textsc{MC}(\{ d(\mathbf{H}) \mid \mathbf{H} \in \mathcal{H} \},\mathcal{FO}(\forall,\exists,\wedge))$.  
The $\textsc{HittingSet}$ problem, known to be $\textup{W}[2]$-hard \cite{FlumGrohe06}, is the problem, 
given a pair $(\mathbf{H},k)$ where $\mathbf{H} \in \mathcal{H}$ and $k \in \mathbb{N}$, 
whether there exists $V \subseteq H$ such that $|V|=k$ and $V \cap U^{\mathbf{H}} \neq \emptyset$ for all $U \in \sigma$; 
$V$ is called a \emph{hitting set} of size $k$ of the hypergraph $\mathbf{H}$.  

Let $(\mathbf{H},k)$ be an instance of $\textsc{HittingSet}$.  
We reduce to the instance $(\mathbf{P},\phi_k)$ of $\textsc{MC}(\mathcal{P},\mathcal{FO}(\forall,\exists,\wedge))$, 
where $\mathbf{P}=d(\mathbf{H})$ and $\phi_k$ is as in (\ref{eq:posetcompl}).  
We check the correctness of the reduction (the complexity is clear).

We claim that $\mathbf{H}$ does not have a hitting set of size $k$ 
if and only if 
$\mathbf{P} \models \phi_k$.  
For the backward direction, 
by Proposition~\ref{prop:cpprop}, 
every choice of $k$ elements $h_1,\ldots,h_k \in \min(\mathbf{P})$ 
have a common upper bound $U \in \max(\mathbf{P})$.  
By construction, $h_i \not\in U^\mathbf{H}$ for all $i \in [k]$, 
that is, $\{h_1,\ldots,h_k\}$ is not a hitting set of $\mathbf{H}$.  
Thus, $\mathbf{H}$ has no hitting sets of size $k$.  For the forward direction, 
we prove the contrapositive.  Assume $\mathbf{P} \not\models \phi_k$.  By Proposition~\ref{prop:cpprop}, 
let $h_1,\ldots,h_k \in \min(\mathbf{P})$ be such that 
no $U \in P$ is a common upper bound of $h_1,\ldots,h_k$.  Let $U \in \max(\mathbf{P})$. Then 
there exists $i \in [k]$ such that $h_i \not\prec^{\mathbf{P}} U$.  
Thus, by construction, $h_i \in U^\mathbf{H}$.  Summarizing, 
for all $U \in \sigma$, there exists $i \in [k]$ such that 
$h_i \in U^\mathbf{H}$, that is, $\{h_1,\ldots,h_k\}$ is a hitting set of $\mathbf{H}$.

The case $r=c$ is proved along the lines of the case $r=d$. 
\shortversion{\qed}\end{proof}}

\longversion{\pfpropdepthh}


%

\section{Conclusion}\label{sect:concl}

We provided a parameterized complexity classification of 
the problem of model checking quantified conjunctive queries 
on posets with respect to the invariants 
in Figure \ref{fig:classification}; in particular, 
we push the tractability frontier of the model checking problem on bounded width posets 
closer towards the full first-order logic.  The question of whether first-order logic 
is fixed-parameter tractable on bounded width posets remains open.


\longversion{
We propose two research questions in classical complexity.  
First, determine the exact complexity of model checking quantified conjunctive queries 
on finite posets; by inspection of the proofs of our hardness results (Theorem~\ref{th:exprhard} and Proposition~\ref{prop:allhardnessresults}), already the $\forall^* \exists^*$ fragment of 
the problem is $\mathrm{NP}$-hard and $\mathrm{coNP}$-hard.  Second, we observed that the problem is polynomial-time tractable 
on certain posets 
(including for instance semilattices, see Corollary~\ref{cor:bwtract}) 
and hard on the bowtie poset (Theorem~\ref{th:exprhard}); these results can be phrased in terms of the 
quantified constraint satisfaction problem $\textsc{QCSP}(\mathbf{A})$, for a suitable template $\mathbf{A}$; 
it would be interesting to revisit (and possibly cover) them in the algebraic 
framework developed for the $\textsc{QCSP}$ \cite{BornerBulatovChenJeavonsKrokhin09}.
}

\subsubsection*{Acknowledgments.}  This research was supported by the European Research Council (Complex Reason, 239962) 
and the FWF Austrian Science Fund (Parameterized Compilation, P26200 and X-TRACT, P26696).

\end{document}

%% file: overwiewnotop.pspdftex
\begin{picture}(0,0)%
\includegraphics{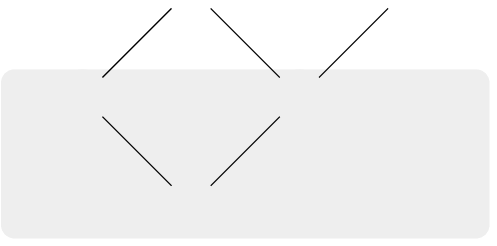}%
\end{picture}%
\setlength{\unitlength}{2279sp}%
\begingroup\makeatletter\ifx\SetFigFont\undefined%
\gdef\SetFigFont#1#2#3#4#5{%
  \reset@font\fontsize{#1}{#2pt}%
  \fontfamily{#3}\fontseries{#4}\fontshape{#5}%
  \selectfont}%
\fi\endgroup%
\begin{picture}(4584,2464)(2014,-10392)
\put(2521,-9106){\makebox(0,0)[lb]{\smash{{\SetFigFont{9}{8.4}{\rmdefault}{\mddefault}{\updefault}{\color[rgb]{0,0,0}$\textup{width}$}%
}}}}
\put(3421,-10006){\makebox(0,0)[lb]{\smash{{\SetFigFont{9}{8.4}{\rmdefault}{\mddefault}{\updefault}{\color[rgb]{0,0,0}$\textup{size}$}%
}}}}
\put(4186,-9106){\makebox(0,0)[lb]{\smash{{\SetFigFont{9}{8.4}{\rmdefault}{\mddefault}{\updefault}{\color[rgb]{0,0,0}$\textup{degree}$}%
}}}}
\put(2971,-8206){\makebox(0,0)[lb]{\smash{{\SetFigFont{9}{8.4}{\rmdefault}{\mddefault}{\updefault}{\color[rgb]{0,0,0}$\textup{cover\textup{-}degree}$}%
}}}}
\put(5041,-8206){\makebox(0,0)[lb]{\smash{{\SetFigFont{9}{8.4}{\rmdefault}{\mddefault}{\updefault}{\color[rgb]{0,0,0}$\textup{depth}$}%
}}}}
\end{picture}%

%% file: exprhard.pspdftex
\begin{picture}(0,0)%
\includegraphics{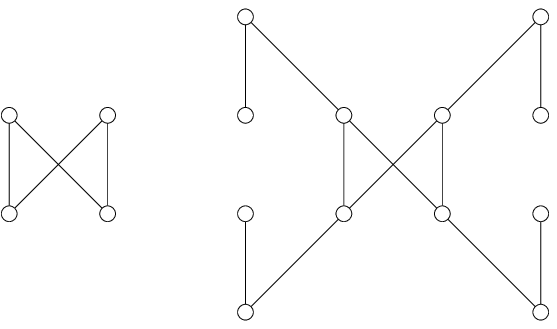}%
\end{picture}%
\setlength{\unitlength}{1657sp}%
\begingroup\makeatletter\ifx\SetFigFont\undefined%
\gdef\SetFigFont#1#2#3#4#5{%
  \reset@font\fontsize{#1}{#2pt}%
  \fontfamily{#3}\fontseries{#4}\fontshape{#5}%
  \selectfont}%
\fi\endgroup%
\begin{picture}(6728,4330)(-1229,-10733)
\put(-1214,-7756){\makebox(0,0)[lb]{\smash{{\SetFigFont{9}{6.0}{\rmdefault}{\mddefault}{\updefault}{\color[rgb]{0,0,0}$3$}%
}}}}
\put(-1214,-9450){\makebox(0,0)[lb]{\smash{{\SetFigFont{9}{6.0}{\rmdefault}{\mddefault}{\updefault}{\color[rgb]{0,0,0}$0$}%
}}}}
\put(-44,-7756){\makebox(0,0)[lb]{\smash{{\SetFigFont{9}{6.0}{\rmdefault}{\mddefault}{\updefault}{\color[rgb]{0,0,0}$1$}%
}}}}
\put(-44,-9450){\makebox(0,0)[lb]{\smash{{\SetFigFont{9}{6.0}{\rmdefault}{\mddefault}{\updefault}{\color[rgb]{0,0,0}$2$}%
}}}}
\put(3646,-7711){\makebox(0,0)[lb]{\smash{{\SetFigFont{9}{6.0}{\rmdefault}{\mddefault}{\updefault}{\color[rgb]{0,0,0}$w_1$}%
}}}}
\put(3646,-9450){\makebox(0,0)[lb]{\smash{{\SetFigFont{9}{6.0}{\rmdefault}{\mddefault}{\updefault}{\color[rgb]{0,0,0}$w_2$}%
}}}}
\put(2566,-7711){\makebox(0,0)[lb]{\smash{{\SetFigFont{9}{6.0}{\rmdefault}{\mddefault}{\updefault}{\color[rgb]{0,0,0}$w_3$}%
}}}}
\put(2566,-9450){\makebox(0,0)[lb]{\smash{{\SetFigFont{9}{6.0}{\rmdefault}{\mddefault}{\updefault}{\color[rgb]{0,0,0}$w_0$}%
}}}}
\put(1441,-10501){\makebox(0,0)[lb]{\smash{{\SetFigFont{9}{6.0}{\rmdefault}{\mddefault}{\updefault}{\color[rgb]{0,0,0}$x_0$}%
}}}}
\put(1441,-8251){\makebox(0,0)[lb]{\smash{{\SetFigFont{9}{6.0}{\rmdefault}{\mddefault}{\updefault}{\color[rgb]{0,0,0}$y_3$}%
}}}}
\put(1441,-6586){\makebox(0,0)[lb]{\smash{{\SetFigFont{9}{6.0}{\rmdefault}{\mddefault}{\updefault}{\color[rgb]{0,0,0}$x_3$}%
}}}}
\put(1441,-8836){\makebox(0,0)[lb]{\smash{{\SetFigFont{9}{6.0}{\rmdefault}{\mddefault}{\updefault}{\color[rgb]{0,0,0}$y_0$}%
}}}}
\put(4771,-10501){\makebox(0,0)[lb]{\smash{{\SetFigFont{9}{6.0}{\rmdefault}{\mddefault}{\updefault}{\color[rgb]{0,0,0}$x_2$}%
}}}}
\put(4771,-8836){\makebox(0,0)[lb]{\smash{{\SetFigFont{9}{6.0}{\rmdefault}{\mddefault}{\updefault}{\color[rgb]{0,0,0}$y_2$}%
}}}}
\put(4771,-8251){\makebox(0,0)[lb]{\smash{{\SetFigFont{9}{6.0}{\rmdefault}{\mddefault}{\updefault}{\color[rgb]{0,0,0}$y_1$}%
}}}}
\put(4771,-6586){\makebox(0,0)[lb]{\smash{{\SetFigFont{9}{6.0}{\rmdefault}{\mddefault}{\updefault}{\color[rgb]{0,0,0}$x_1$}%
}}}}
\end{picture}%

%% file: rss-mfcs14.bbl
\begin{thebibliography}{10}

\bibitem{BovaGanianSzeider14}
S.~Bova, R.~Ganian, and S.~Szeider. 
\newblock Model Checking Existential Logic on Partially Ordered Sets.
\newblock In {\em CSL-LICS}, 2014.  Preprint in {\em CoRR}, abs/1405.2891, 2014.

\shortversion{\bibitem{BovaGanianSzeiderIPEC14}
S.~Bova, R.~Ganian, and S.~Szeider. 
\newblock Quantified Conjunctive Queries on Partially Ordered Sets.
Preprint in {\em CoRR}, TO DO, 2014.}

\bibitem{BornerBulatovChenJeavonsKrokhin09}
F. B{\" o}rner, A. Bulatov, H. Chen, P. Jeavons, and A. Krokhin.
\newblock The Complexity of Constraint Satisfaction Games and QCSP.
\newblock {\em Inform.\ Comput.}, 207(9), 923--944, 2009.

\bibitem{CaspardLeclercMonjardet12}
N.~Caspard, B.~Leclerc, and B.~Monjardet.
\newblock {\em Finite Ordered Sets}.
\newblock Cambridge University Press, 2012.

\bibitem{ChenDalmau12}
H.~Chen and V.~Dalmau.
\newblock Decomposing Quantified Conjunctive (or Disjunctive) Formulas. 
\newblock In {\em LICS}, 2012.

\bibitem{Courcelle90recognizable}
B.~Courcelle.
\newblock The Monadic Second-Order Logic of Graphs. I. Recognizable Sets of Finite Graphs.
\newblock {\em Inform.\ Comput.}, 85(1):12--75, 1990.

\bibitem{CourcelleMakowskyRotics00}
B.~Courcelle, J.~A. Makowsky, and U.~Rotics.
\newblock Linear Time Solvable Optimization Problems on Graphs of Bounded Clique-Width.
\newblock {\em Theory Comput. Syst.}, 33(2):125--150, 2000.

\bibitem{FlumGrohe06}
J.~Flum and M.~Grohe. 
\newblock {\em Parameterized Complexity Theory}. 
\newblock Springer, 2006.

\bibitem{GrahamGrotschelLovasz95}
R.~L. Graham, M.~Gr\"{o}tschel, and L.~Lov\'{a}sz (editors).
\newblock {\em Handbook of Combinatorics, Vol.~1}.
\newblock MIT Press, 1995.

\bibitem{Grohe07}
M.~Grohe.
\newblock The Complexity of Homomorphism and Constraint Satisfaction Problems seen from the Other Side.
\newblock {\em J. of the ACM}, 54(1), 2007.

\bibitem{GroheKreutzer11}
M.~Grohe and S.~Kreutzer. 
\newblock Methods for Algorithmic Meta Theorems.
\newblock In {\em Model Theoretic Methods in Finite Combinatorics}, pp.~181--206. AMS, 2011.

\bibitem{GroheKreutzerSiebertz14}
M.~Grohe, S.~Kreutzer, and S.~Siebertz. 
\newblock Deciding First-Order Properties of Nowhere Dense Graphs.  
\newblock In {\em STOC}, 2014.  Preprint in {\em CoRR}, abs/1311.3899, 2013.

\bibitem{GroheSchwentickSegoufin01}
M.~Grohe, T.~Schwentick, and L.~Segoufin.
\newblock When is the Evaluation of Conjunctive Queries Tractable?
\newblock In {\em STOC}, 2001.

\bibitem{NesetrilOssonadeMendez12}
J.~Ne{\v s}et{\v r}il and P.~Ossona de Mendez. 
\newblock {\em Sparsity}. 
\newblock Springer, 2012.

\bibitem{PrattTiuryn96}
V.~R. Pratt and J.~Tiuryn.
\newblock Satisfiability of Inequalities in a Poset. 
\newblock {\em Fund. Inform.}, 28(1-2):165--182, 1996.

\longversion{\bibitem{Schroder03}
B.~Schr{\"o}der.
\newblock {\em Ordered Sets: An Introduction}.
\newblock Birkh{\"a}user, 2003.}

\bibitem{Seese96}
D.~Seese.
\newblock Linear Time Computable Problems and First-Order Descriptions.
\newblock {\em Math.\ Struct.\ in Comp.\ Science}, 6(6):505--526, 1996.

\end{thebibliography}
